\documentclass[cmp, runningheads, envcountsect, envcountsame, envcountreset]{svjour}

\journalname{}

\usepackage{graphicx}
\usepackage{epstopdf}
\usepackage{mathrsfs}
\usepackage{color}
\usepackage{array}
\usepackage[abbrev]{amsrefs}
\usepackage{fancyhdr}
\usepackage{amsmath}
\usepackage{amsfonts,amssymb}
\usepackage{cancel}
\usepackage{bbm}
\usepackage[bookmarks=false,pdfstartview={FitH}]{hyperref}
\usepackage{url}
\usepackage{tikz}
\usepackage{enumitem}
\usetikzlibrary{calc,intersections,through,backgrounds}
\usepackage[FIGTOPCAP]{subfigure}
\usepackage{tikz}

\usetikzlibrary{calc}

\tikzset{bkgd_line/.style={very thin,color=gray}}
\tikzset{dashed_bkgd_line/.style={thin, dashed,color=gray}}

\tikzset{dimension_line/.style={very thin,color=gray}}
\tikzset{axis_line/.style={thin}}
\tikzset{object_line/.style={thick}}
\tikzset{textlabel/.style={color=black}}

\DefineJournal{aop}{0091-1798}{Ann. Probab.}{The Annals of Probability}
\DefineJournal{cmp}{0010-3616}{Comm. Math. Phys.}{Communications in Mathematical Physics}
\DefineJournal{ptrf}{0178-8051}{Probab. Theory Related Fields}{Probability Theory \& Related Fields}
\DefineJournal{im}{0020-9910}{Invent. Math.}{Inventiones Mathematicae}
\DefineJournal{prl}{0031-9007}{Phys. Rev. Lett.}{Physical Review Letters}
\DefineJournal{pra}{1050-2947}{Phys. Rev. A}{Physical Review. A}
\DefineJournal{jsp}{0022-4715}{J. Statist. Phys.}{Journal of Statistical Physics}
\DefineJournal{jfa}{0022-1236}{J. Funct. Anal.}{Journal of Functional Analysis}
\DefineJournal{ida}{0219-0257}{Infin. Dimens. Anal. Quantum Probab. Relat. Top.}{Infinite Dimensional Analysis, Quantum Probability and Related Topics}
\DefineJournal{cjm}{0008-414X}{Canad. J. Math.}{Canadian Journal of Mathematics}
\DefineJournal{jems}{1435-9855}{J. Eur. Math. Soc.}{Journal of the European Mathematical Society}
\DefineJournal{aihpps}{0246-0203}{Ann. Inst. Henri Poincar\'{e} Probab. Stat.}{Annales de l'Institut Henri Poincar\'{e} Probabilit\'{e}s et Statistiques}
\DefineJournal{expmath}{1058-6458}{Exp. Math.}{Experimental Mathematics}
\DefineJournal{ma}{0025-5831}{Math. Ann.}{Mathematische Annalen}

\title{Statistical mechanics of Bose gas in Sierpinski carpets}
\author{Joe P. Chen}

\institute{Departments of Physics and Mathematics, Cornell University, Ithaca, NY 14853.\\  \email{joe.p.chen@cornell.edu}}

\date{\today}

\numberwithin{equation}{section}

\begin{document}


\begin{abstract}

We carry out a mathematically rigorous investigation into the equilibrium thermodynamics of massless and massive bosons confined in generalized Sierpinski carpets (GSCs), a class of infinitely ramified fractals having non-integer Hausdorff dimensions $d_h$. Due to the anomalous walk dimension $d_w>2$ associated with Brownian motion on GSCs, all extensive thermodynamic quantities are shown to scale with the spectral volume with dimension $d_s = 2(d_h/d_w)$ rather than the Hausdorff volume. We prove that for a low-temperature, high-density ideal massive Bose gas in an unbounded GSC, Bose-Einstein condensation occurs if and only if $d_s>2$, or equivalently, if the Brownian motion on the GSC is transient. We also derive explicit expressions for the energy of blackbody radiation in a GSC, as well as the Casimir pressure on the parallel plate of a fractal waveguide modelled after a GSC. Our proofs involve extensive use of the spectral zeta function, obtained via a sharp estimate of the heat kernel trace. We believe that our results can be verified through photonic and cold atomic experiments on fractal structures.

\end{abstract}

\maketitle

\section{Introduction}

This paper is devoted to the study of the thermodynamics of quantum gases in fractal spaces. It was long recognized, by H. Lorentz and H. Weyl~\cite{Weyl1912}, that a deep connection exists between the thermodynamic properties of quantum gases and the underlying spectral geometry. One can probe the asymptotic behavior of elliptic linear differential operators in a given space by measuring the energy or pressure of a quantum gas in the said space. Since then much of this work has been done on Euclidean spaces. Here we wish to address the following question: 

\emph{How does the fractal geometry affect the laws of quantum many-body physics? Conversely, what information about the fractal geometry can we obtain using probes made of quantum particles?}

Below are some specific examples which are familiar to students of quantum and statistical physics, and which have wide ramifications in condensed matter physics and quantum field theory.
\begin{itemize}
\item Does Bose-Einstein condensation of atomic gas occur in non-integer dimension? If so, what dimension is it? And what is the critical density for condensation?
\item What is the analog of the Stefan-Boltzmann law for blackbody radiation when the blackbody itself is a fractal?
\item At zero temperature, what is the Casimir pressure produced by vacuum fluctuations inside a fractal? 
\end{itemize}

We will answer all three questions in the case where the fractal is a \emph{generalized Sierpinski carpet (GSC)}, whose representatives are the standard two-dimensional Sierpinski carpet and the three-dimensional Menger sponge (Fig.~\ref{fig:GSC}). These fractals have connected interior, are highly symmetric, and most of all, are infinitely ramified, which makes the analysis difficult. Conventional analytic tools, such as Fourier transform or spectral decimation, no longer apply. Essentially all the rigorous results known today originate from the study of Brownian motion. On the other hand, GSCs are embeddable in Euclidean space and resemble more realistic fractals found in nature. So we believe that a careful analysis of quantum gases in GSCs is warranted, because it provides us an avenue of attacking quantum many-body problems on general irregular spaces. 

Our exposition is aimed at both physicists and mathematicians. It serves a dual purpose: to highlight the latest developments from the mathematical analysis on fractals, which have just begun to percolate through the physics community; and to illustrate how state-of-the-art potential theoretic results can be applied to answer physically inspired problems.

\textbf{Notations.} We write $C$ and $c$ for positive constants which may change from line to line. If a constant has a specific value, then we will add a numeral subscript to indicate this, \emph{e.g.} $C_1$. Given two real-valued functions $f$ and $g$, we say that $f$ is comparable to $g$ if there exist constants $c, C>0$ such that $c g(x) \leq f(x) \leq C g(x)$, denoted by $f(x) \asymp g(x)$ for short. For any $A \subset \mathbb{R}^d$ and $L>0$, we write $LA=\{x \in \mathbb{R}^d : x/L \in A\}$. Finally, we use $\mathbb{N}_0$ and $\mathbb{N}$ to denote, respectively, the set of natural numbers with $0$ and without $0$.

\subsection{Generalized Sierpinski carpet}

Let $F_0:=[0,1]^d$ be the unit cube in $\mathbb{R}^d$, $d\geq 2$. Fix a \emph{length scale factor} $l_F \in \mathbb{N}$, $l_F \geq 3$, and let $\mathcal{S}_n$ be the collection of closed cubes of side $l_F^{-n}$ with vertices in $l_F^{-n} \mathbb{Z}^d$. For $A \subset \mathbb{R}^d$, let $\mathcal{S}_n(A) = \{S\in \mathcal{S}_n : S \subset A \}$. Denote by $\Psi_S$ the orientation-preserving affine map which maps $F_0$ to $S \in \mathcal{S}_n$.

\begin{figure}
\centering
\subfigure[~]{
\includegraphics[width=0.32\textwidth]{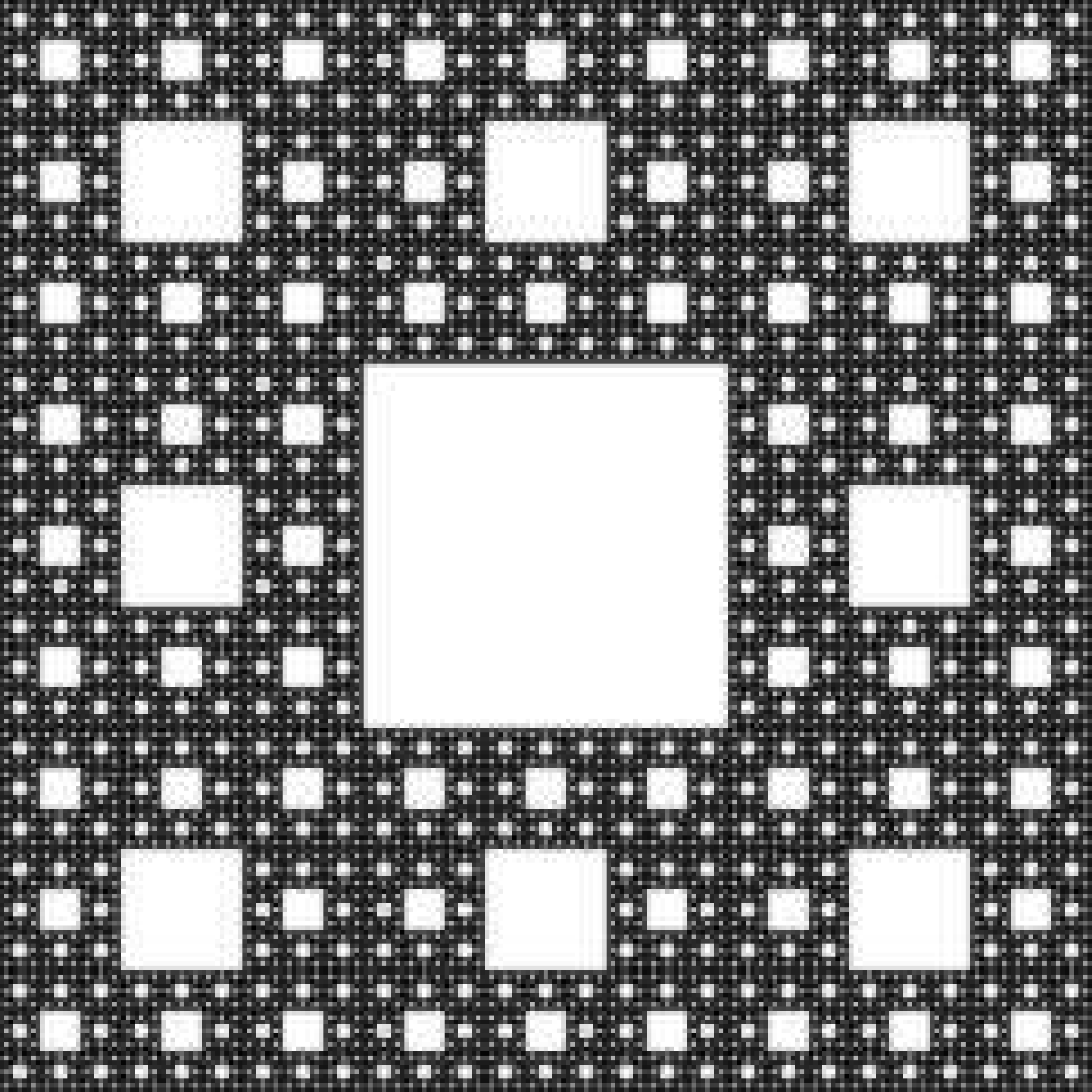}
\label{subfig:SC}
}
\subfigure[~]{
\includegraphics[width=0.3\textwidth]{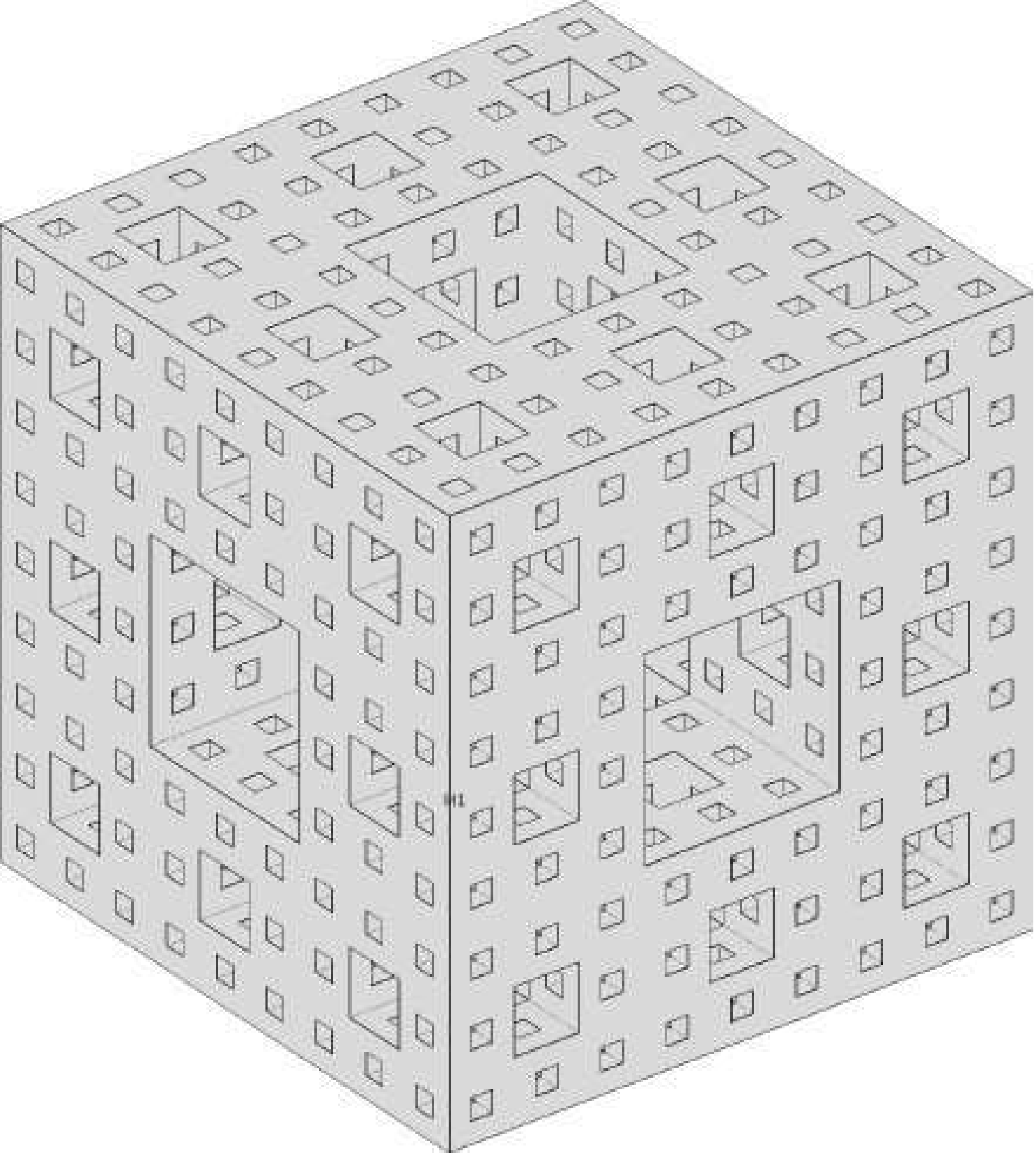}
\label{subfig:MS}
}
\caption{Examples of generalized Sierpinski carpet.~\subref{subfig:SC} The standard Sierpinski carpet $SC(3,1)$, with $l_F=3$ and $m_F=8$.~\subref{subfig:MS} The Menger sponge $MS(3,1)$, with $l_F=3$ and $m_F=20$.}
\label{fig:GSC}
\end{figure}

Introduce a decreasing sequence $\{F_n\}_n$ of closed subsets of $F_0$ as follows. Fix $m_F \in \mathbb{N}$, $1\leq m_F < l_F^d$, and let $F_1$ be the union of $m_F$ distinct elements of $\mathcal{S}_1(F_0)$. In other words, we construct $F_1$ by retaining $m_F$ of the cells of length $l_F^{-1}$, and removing the rest. We will refer to $m_F$ as the \emph{mass scale factor}. Then by iteration we let
$$
F_{n+1} = \bigcup_{S\in \mathcal{S}_n(F_n)} \Psi_S(F_1)  = \bigcup_{S \in \mathcal{S}_1(F_1)}\Psi_S(F_n)~,~n\geq 1.
$$
This iterated function system of contractions $\{\Psi_S\}$ has a unique fixed point $F= \bigcap_{n=0}^\infty F_n$. By standard arguments, the Hausdorff dimension of $F$ is $d_h(F)= \log m_F/\log l_F$. 

\begin{definition}
We say that $F$ is a \emph{generalized Sierpinski carpet (GSC)} if the following conditions on $F_1$ hold:
\begin{enumerate}[label={(H\arabic*)},nolistsep]
\item (Symmetry) $F_1$ is preserved under the isometries of the unit cube. \label{cond:H1}
\item (Connectedness) ${\rm Int}(F_1)$ is connected, and contains a path connecting the hyperplanes $\{x_1=0\}$ and $\{x_1=1\}$. \label{cond:H2}
\item (Non-diagonality) Let $B$ be a cube in $F_0$ which is the union of $2^d$ distinct elements of $\mathcal{S}_1$. Then if ${\rm Int}(F_1\cap B) \neq \emptyset$, it is connected. \label{cond:H3}
\item (Borders included) $F_1$ contains the segment $\{x \in \mathbb{R}^d: x_1\in [0,1],~x_2=\cdots=x_d=0\}$. \label{cond:H4}
\end{enumerate}
\label{def:GSC}
\end{definition}

We will denote by $\partial F$ the boundary of $F$, and by $\partial_o F = \partial([0,1]^d) \cap \partial F$ the outer boundary.

In order to discuss the thermodynamic limit we need to consider unbounded versions of GSC. For each GSC $F$, we call $\tilde F:=\bigcup_{n=0}^\infty l_F^n F_n$ the corresponding \emph{pre-carpet}, and $F_\infty := \bigcup_{n=0}^\infty l_F^n F$ the \emph{unbounded carpet}. Plainly speaking, the former is the infinite "blow-up" of the fractal, whereas the latter has both "blow-up" and "blow-down." Our focus will be on the unbounded carpet, where sharp results can be stated, though we believe many of them carry over to the pre-carpet modulo minor modifications.

\subsection{Statement of the problem} 

Let $F$ be a GSC, $\nu$ be the self-similar measure on $F$, and $\mathcal{H}_1 := L^2(F,\nu)$ be the single-particle Hilbert space. Given a self-adjoint, bounded-below Hamiltonian operator $H:\mathcal{H}_1\to\mathcal{H}_1$, which captures the physics of a certain ideal Bose gas inside $F$, we wish to compute its grand canonical partition function $\Xi_{\beta,\mu}$ at inverse temperature $\beta \in (0,\infty]$ and chemical potential $\mu \in (-\infty,\inf{\rm Spec}(H)]$:
\begin{equation}
\log \Xi_{\beta,\mu}  = -{\rm Tr}_{\mathcal{H}_1}\log\left(1-e^{-\beta(H-\mu)}\right).
\label{eq:GCPF}
\end{equation}
Note that~(\ref{eq:GCPF}) already takes into account the Bose-Einstein statistics satisfied by the particles (see Section~\ref{sec:QSM}).

The Hamiltonian considered in this paper involves the (nonnegative) Laplacian $-\Delta$ on the measure space $(F,\nu)$, constructed via either the Barlow-Bass approach (Brownian motions on the outer approximations) or the Kusuoka-Zhou approach (random walks on the associated graphs). Up to time change, both Laplacians generate the same (and the unique) Brownian motion which is invariant under the isometries of the carpet, as~\cite{BBKT} has shown. We will concentrate on the following quantum gases, working in units where Planck's constant ($\hbar$), and either the speed of light or twice the atomic mass, are $1$. 
\begin{itemize}
\item \textbf{Massive Bose gas} (\emph{e.g.} atoms), where $H =-\Delta$ and $\mu \leq \inf {\rm Spec} (H)$. The associated scalar field $\psi$ satisfies the usual Schr\"{o}dinger equation $i\partial_t \psi = -\Delta\psi$.
\item \textbf{Massless Bose gas} (\emph{e.g.} photons), where $H =\sqrt{-\Delta}$ and $\mu = 0$. Recall that the massless scalar field $\psi$ satisfies the wave equation $-\partial_{tt}\psi = (-\Delta) \psi$. Since the Laplacian $-\Delta$ is a nonnegative self-adjoint operator, it has a unique square root $\sqrt{-\Delta}$, also known as the Dirac operator. Thus we may take the formal square root of the wave equation to obtain a 1st-order-in-time PDE $i\partial_t \psi = \sqrt{-\Delta}\psi$, which is equivalent to a Schr\"{o}dinger equation $i\partial_t \psi = H\psi$ with $H=\sqrt{-\Delta}$. 
\end{itemize}

\subsection{Methodology}

An immediate obstacle to a rigorous thermodynamic calculation on GSCs is the lack of knowledge about the exact spectrum of the Laplacian. In fact, even the eigenvalue counting function (or the integrated density of states in physics parlance) in its optimal form cannot be proved (see Remark~\ref{re:ECF}). Fortunately, many efforts which went into studying the short-time asymptotics of the heat kernel on GSC have produced sharp enough estimates, to the point that we can show that the corresponding spectral zeta function $\zeta_\Delta(s,\gamma)= {\rm Tr}_{\mathcal{H}_1}(-\Delta+ \gamma)^{-s}$ admits a meromorphic extension (Section~\ref{sec:speczeta}). This zeta function is then used effectively in computing the thermodynamic partition function, among other things.

We stress that the zeta function technology used on fractals applies equally well to Euclidean domains, manifolds, and graphs, and \emph{a fortiori} produces the usual results on $\mathbb{R}^d$ or $\mathbb{Z}^d$. Furthermore, it unifies the treatments of massless and massive Bose gases: see~\cite{kirsten2010} for a nice exposition in the $\mathbb{R}^d$ case.

\subsection{Main results}

Below is a summary of the major results on the thermodynamics of massive and massless Bose gases in Sierpinski carpets. All terminology will be explained in subsequent sections.

\begin{itemize}
\item \underline{Bose-Einstein condensation (Section~\ref{sec:atoms}).} For a low-temperature, high-density ideal massive Bose gas in an unbounded GSC $F_\infty$, Bose-Einstein condensation occurs at positive temperature if and only if the \emph{spectral dimension} of the carpet $d_s(F_\infty) > 2$. In this case, whenever the Bose gas density exceeds $\overline\rho_c(\beta) := \frac{C_1(F_\infty)}{(4\pi\beta)^{d_s/2}} \zeta\left(\frac{d_s}{2}\right)$, where $C_1(F_\infty) \geq 1$ is a constant depending on $F_\infty$ (see Proposition~\ref{prop:uppdensity}), then any excess density must condense in the lowest eigenfunction of the Laplacian. See Theorems~\ref{thm:BEC} and~\ref{thm:equivBEC}.

\item \underline{Blackbody radiation (Section~\ref{sec:photon}).} Let $F$ be a GSC. At inverse temperature $\beta$, the energy per unit spectral volume of photons inside $LF$ is 
$$
\mathcal{E}(\beta,L)= \beta^{-(d_s(F)+1)} H_1 \left(-\log\left(\frac{\beta}{2L}\right)\right) + o(1)~~~\mbox{as}~~L\to\infty,
$$
where $H_1$ is a periodic function of period $\frac{1}{2}d_w(F)\log l_F$ given in Proposition~\ref{prop:BBEnergy}, and $d_w(F)$ is the \emph{walk dimension} of the carpet.

\item \underline{Casimir effect (Section~\ref{sec:Casimir}).} Consider a waveguide $\Omega_{a,b}=aF \times [0,b]$ modelled after a 2-dimensional GSC $F$, and impose Dirichlet conditions on the outer boundary. If one places a pair of parallel plates at $aF\times \{0\}$ and $aF \times \{b\}$, respectively, then the zero-temperature Casimir pressure on each plate is given by
$$
P_{\rm Cas}(a,b) = b^{-(d_s(\Omega)+1)} H_2 \left(-\log\left(\frac{b}{a}\right)\right) + o(1)~~~\mbox{as}~~a \to \infty,
$$
where $d_s(\Omega)=d_s(F)+1$ is the spectral dimension of the waveguide $\Omega_{a,b}$, and $H_2$ is a periodic function of period $\frac{1}{2}d_w(F)\log l_F$ given in Proposition~\ref{prop:CasimirPressure}.
\end{itemize}

These are to be compared with the classical "textbook" results in Euclidean space:

\begin{itemize}
\item For a low-temperature, high-density ideal massive Bose gas in $\mathbb{R}^d$ or $\mathbb{Z}^d$, Bose-Einstein condensation occurs at positive temperature if and only if $d\geq 3$. In this case, whenever the Bose gas density exceeds $\rho_c(\beta):= \frac{1}{(4\pi\beta)^{d/2}} \zeta\left(\frac{d}{2}\right)$, any excess density must condense in the lowest eigenfunction of the Laplacian.

\item At inverse temperature $\beta$, the energy per unit Euclidean volume of photons in a cube $[0,L]^d$ is
$$
\mathcal{E}(\beta,L) =\frac{1}{\beta^{d+1}} \frac{d}{\pi^{(d+1)/2}}\Gamma\left(\frac{d+1}{2}\right) \zeta(d+1) + o(1)~~~\mbox{as}~~L\to\infty.
$$
When $d=3$ we recover the familiar "$T^4$ law," $\mathcal{E}(\beta,L) = \pi^2/(30\beta^4) + o(1)$.

\item Consider the rectangular waveguide $[0,a]^2 \times [0,b]$ with Dirichlet boundary conditions, with $[0,a]^2 \times\{0\}$ and $[0,a]^2 \times \{b\}$ being the two parallel plates. Then the zero-temperature Casimir pressure on each plate is given by
$$
P_{\rm Cas}(a,b) = -\frac{\pi^2}{240 b^4} + o(1)~~~\mbox{as}~~a\to\infty.
$$
\end{itemize}

While our results are stated for true fractals, we believe that many of these behaviors can already be seen on finite-level approximations of the fractal, as indicated by our numerical work on the spectrum of the Laplacian (J.P. Chen and R.S. Strichartz, preprint [CS], M. Begu\'{e}, T. Kalloniatis and R.S. Strichartz, arXiv:1201.5136 [BKS]). It would therefore be edifying if experimentalists can take on the challenge of constructing fractal-based quantum systems using, \emph{e.g.} metamaterials, optical lattices, or superconducting qubits (for instance the proposal by~\cite{Tsomokos}), and testing our results.

The paper is organized as follows. In Section~\ref{sec:SC} we recapitulate several recently established results from the analysis on GSCs, namely the uniqueness of Brownian motion and the estimate of the heat kernel trace. In Section~\ref{sec:speczeta} we introduce the spectral zeta function on GSCs, show that it admits a meromorphic extension to $\mathbb{C}$, and give its poles and residues. After recalling the rudiments of quantum statistical mechanics in Section~\ref{sec:QSM}, we then present our thermodynamic computations of the massive and massless Bose gases in Sections~\ref{sec:atoms},~\ref{sec:photon} and~\ref{sec:Casimir}. We conclude with Section~\ref{sec:graph} by discussing a special case of interacting Bose gas on Sierpinski carpet graphs, and offering some open problems.

\section{Established results on Sierpinski carpets} \label{sec:SC}

\subsection{Existence and uniqueness of Brownian motion}

In this subsection we give a quick account of how the Laplacian on GSC is constructed. The precise details are highly nontrivial and involve various techniques in potential theory (see \emph{e.g.}~\cite{FOT} for background): we refer the reader to the original literature. For the purposes of this paper, it is enough to recognize the following mathematical facts. A Laplacian $\Delta$ is in 1-to-1 correspondence with a nonnegative symmetric Markovian quadratic form, called the \emph{Dirichlet form} $\mathcal{E}(u,v)=\int u(-\Delta v)$, on an appropriate Banach space. Moreover, the said Laplacian generates a Markov process, which in our setting is either a simple random walk (on a graph) or a Brownian motion (on a subset of $\mathbb{R}^d$).

\textbf{The Barlow-Bass construction}~\cite{BB89,BB90localtimes,BB92,BB99}.
Let $W^n_t$ be a reflecting Brownian motion on the $n$th approximating domain $F_n$ of $F$. In order to produce a Brownian motion on the fractal $F$, one has to use the self-similarity of the carpet and take a suitable scaling limit. Barlow and Bass proved that there exists a family of time-scale factors $\{a_n\}_n$ satisfying 
\begin{equation}
c_1 \left(\frac{\rho_F m_F}{l_F^2}\right)^n \leq a_n \leq c_2 \left(\frac{\rho_F m_F}{l_F^2}\right)^n,
\label{eq:rhoFest}
\end{equation}
for some $c_1, c_2, \rho_F \in (0,\infty)$ independent of $n$, such that the sequence of sped-up Brownian motions $X^n_t=W^n_{a_n t}$ on $F_n$ has a subsequential limit. (The role of $\rho_F$ will be discussed in the next subsection.) Any such limit process $X_t$, which respects the symmetry of $F$ (henceforth referred to as $F$-symmetric), is called a Brownian motion on $F$. We denote its infinitesimal generator by $\mathcal{L}_{BB}$, the Barlow-Bass Laplacian (with Neumann conditions on $\partial_o F$).

One can also rephrase the above result in terms of Dirichlet forms. Let $\nu_n(dx)=(l_F^d/m_F)^n dx$ be the Borel probability measure on $F_n$ which assigns equal weight to each $n$-th level cell of $F_n$. (Note that $\nu_n$ converges weakly to a probability measure $\nu$, which is a constant multiple of the $d_h(F)$-dimensional Hausdorff measure, on $F$.) We introduce the Dirichlet energy on $(F_n, \nu_n)$
\begin{equation}
\mathcal{E}_n(u) = \int_{F_n} \left|\nabla u(x)\right|^2 \nu_n(dx)
\end{equation}
for all $u \in L^2(F_n,\nu_n)$ such that $\mathcal{E}_n(u)<\infty$, and obtain the corresponding Dirichlet form by polarization: $\mathcal{E}_n(u,v) = \frac{1}{4}\left[\mathcal{E}_n(u+v) - \mathcal{E}_n(u-v)\right]$. Then the rescaled Dirichlet energies $\mathscr{E}_n(u) := a_n \mathcal{E}_n(u)$ converge in subsequence to
\begin{equation}
\mathscr{E}_{BB}(u) = \sup_{t>0} \frac{1}{t}\langle (1-T_t) u,u \rangle_{L^2(F,\nu)},
\label{eq:limitingDE}
\end{equation}
where $T_t = e^{t \mathcal{L}_{BB}}$ is the semigroup associated with $X_t$. The corresponding Dirichlet form is strongly local, regular, conservative, $F$-symmetric, and self-similar:
\begin{equation}
\mathscr{E}_{BB}(u,v) = \sum_{S \in S_1(F_1)} \rho_F \mathscr{E}_{BB}(u\circ \Psi_S,v\circ \Psi_S).
\label{eq:SSDE}
\end{equation}

\begin{figure}
\centering
\includegraphics[width=0.35\textwidth]{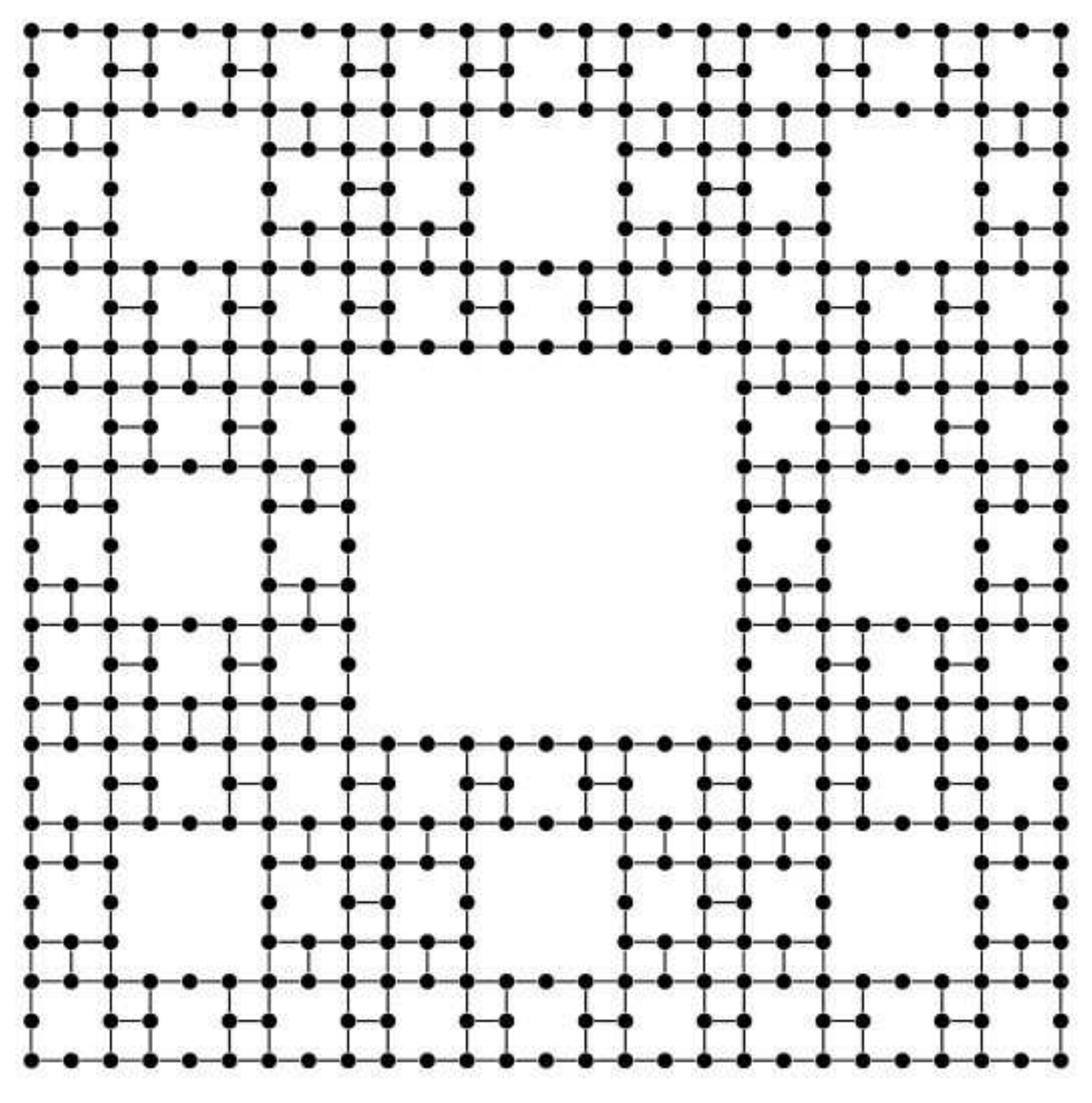}
\caption{The level-3 graph approximation $G_3$ of the standard Sierpinski carpet $SC(3,1)$.}
\label{fig:SC3graph}
\end{figure}

\textbf{The Kusuoka-Zhou construction}~\cite{KusuokaZhou}. Given $F_n$, let $G_n=(V_n,E_n)$ be the graph whose vertices lie at the centers of the cells $S\in \mathcal{S}_n(F_n)$, and whose edges connect vertices in nearest neighboring cells: see Fig.~\ref{fig:SC3graph}. Define the Dirichlet energy to be the usual graph energy
$$\mathcal{E}_n(u) =\sum_{\langle xy \rangle \in E_n} [u(x)-u(y)]^2$$
for all continuous functions $u, v \in C(G_n;\mathbb{R})$, and obtain the Dirichlet form via polarization. By checking a series of geometric conditions for which Poincar\'{e} and Harnack inequalities hold, Kusuoka and Zhou were able to prove that $\{\rho_F^n \mathcal{E}_n\}_n$ converges in subsequence to a strongly local, regular, conservative, $F$-symmetric Dirichlet form $\mathscr{E}_{KZ}$ satisfying the self-similar identity
\begin{equation}
\mathscr{E}_{KZ}(u,v) = \sum_{S \in S_1(F_1)} \rho_F \mathscr{E}_{KZ}(u\circ \Psi_S,v\circ \Psi_S).
\end{equation}
The corresponding Brownian motion has infinitesimal generator $\mathcal{L}_{KZ}$, which we call the Kusuoka-Zhou Laplacian.

A question which lingered for almost two decades was whether the two constructions yield the same limiting Laplacian on $F$. This was settled definitively by Barlow, Bass, Kumagai \& Teplyaev.

\begin{theorem}[{\cite{BBKT}*{Theorem 1.2}}]
Let $F$ be a GSC equipped with the self-similar measure $\nu$. Up to scalar multiples, the set of non-zero, local, regular, conservative, and $F$-symmetric Dirichlet forms on $(F,\nu)$ contains at most one element.
\label{thm:uniq}
\end{theorem}

\begin{figure}
\subfigure{
\includegraphics[scale=0.4]{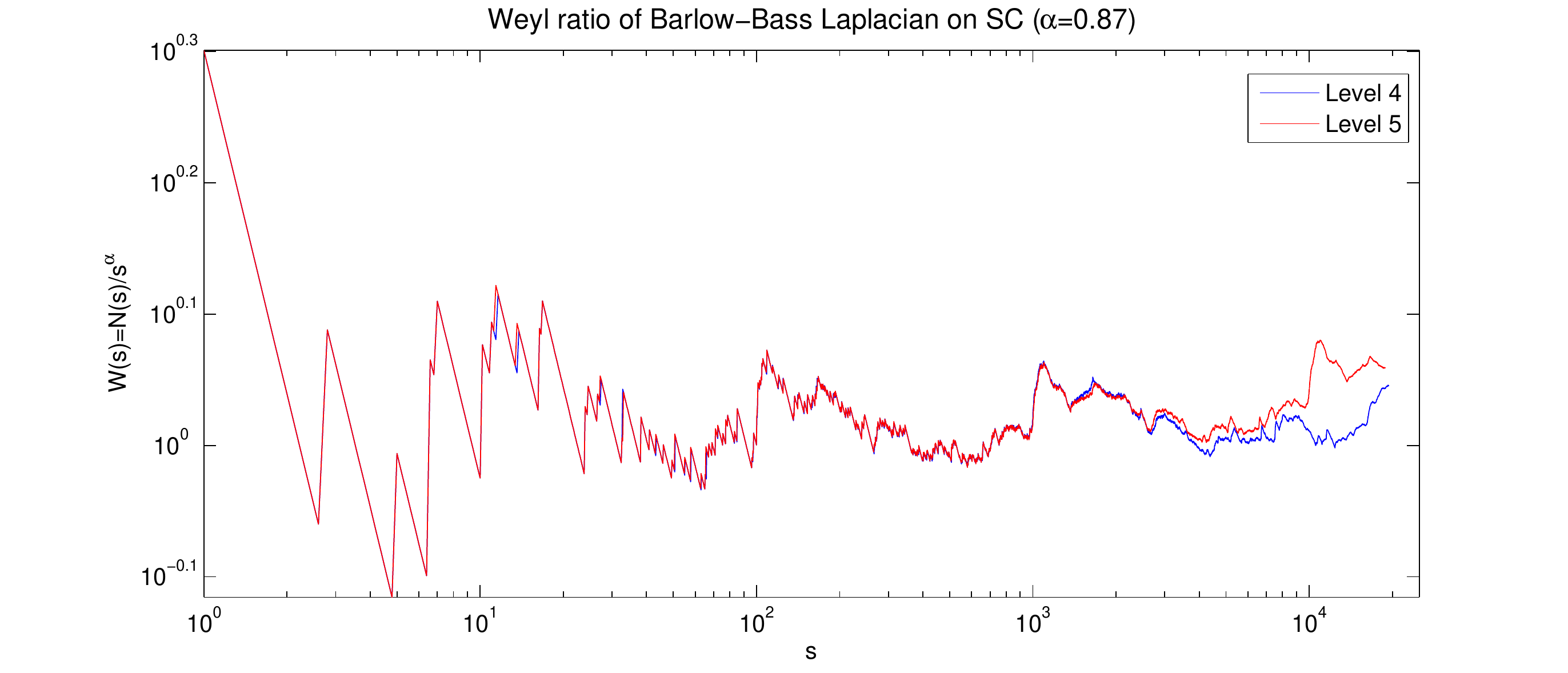}
\label{subfig:BBSC}
}
\subfigure{
\includegraphics[scale=0.4]{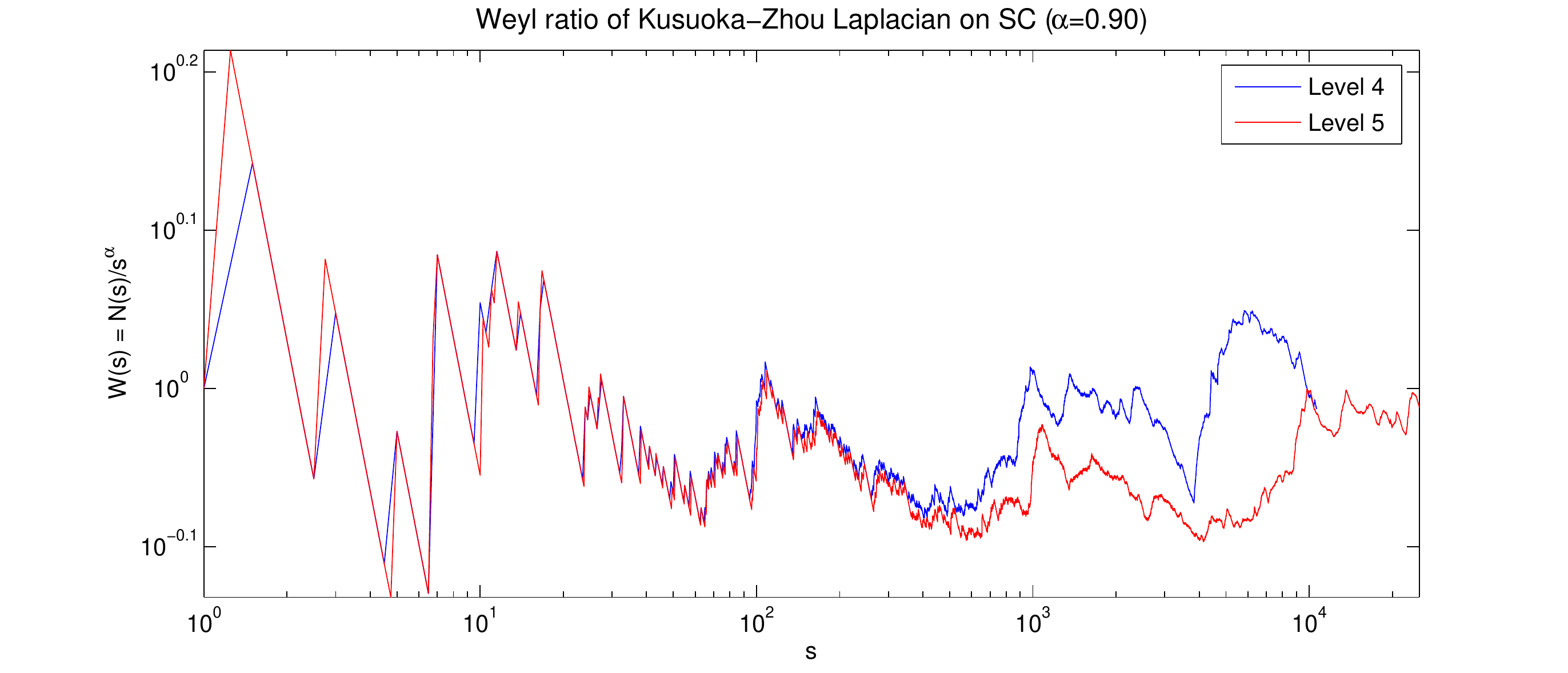}
\label{subfig:KZSC}
}
\caption{Numerical Weyl ratios $W(s) = s^{-d_s/2}N(s)$ associated with~\subref{subfig:BBSC} the Barlow-Bass Laplacian~[CS] and~\subref{subfig:KZSC} the Kusuoka-Zhou Laplacian~[BKS] on approximations $F_4$ and $F_5$ of the standard two-dimensional Sierpinski carpet $F=SC(3,1)$, where $N(s):=\#\{\lambda \in {\rm Spec}(-\Delta): \lambda/\lambda_1 <s\}$ is the eigenvalue counting function of the Neumann Laplacian, normalized by the lowest nonzero eigenvalue $\lambda_1$.}
\label{fig:SCWeyl}
\end{figure}

An equivalent statement to Theorem~\ref{thm:uniq} is that up to deterministic time change, $\mathcal{L}_{BB}$ and $\mathcal{L}_{KZ}$ both generate the unique $F$-symmetric Brownian motion. Henceforth we will denote this unique Laplacian by $\Delta$. 

It should be noted that through extensive numerical computations~[CS,BKS], we can demonstrate that both versions of the Laplacian on \emph{finite} approximations $F_n$ of $F$ coincide at the bottom of the spectrum, as Fig.~\ref{fig:SCWeyl} shows. Notice that upon removing the power-law growth from the eigenvalue counting function, the remainder exhibits logarithmically periodic oscillation, which indicates the fractal nature of the geometry (see Section~\ref{subsec:HKT} for more details). For large enough $n$ (typically $n \geq 3$), the first $(n-2)$ segments of the two spectra, and to a lesser extent the $(n-1)$th segment, agree very well. Then deviation creeps in the $n$th segment: the Barlow-Bass spectrum becomes more characteristic of the spectrum on $[0,1]^d$, while the Kusuoka-Zhou spectrum is truncated due to the finite cardinality $m_F^n$ of the graph.

So despite our inability to state precise results rigorously on the finite approximations $F_n$, we believe strongly that all the asymptotic thermodynamic results presented in this paper, which are stated for $F$ but invariably tied to the bottom of the spectrum, should be visible on $F_n$ for sufficiently large $n$, beginning perhaps from $n=3$.

\subsection{Resistance and heat kernel estimates}

The aforementioned coefficient $\rho_F$ in~(\ref{eq:rhoFest}) is called the \emph{resistance scale factor} of the carpet $F$. When viewing each $F_n$ as an electrical network, $\rho_F$ gives the renormalization factor relating the resistance of $F_n$ to that of $F_{n+1}$. To this date there is no known closed form expression for $\rho_F$. The best known bound, obtained by cutting and shorting resistances, is~\cite{BB99}*{Proposition 5.1}
\begin{equation}
\frac{l_F^2}{m_F} \leq \rho_F \leq 2^{1-d} l_F.
\label{eq:resistanceest}
\end{equation}

The connection between resistance and Brownian motion is as follows. Let $d_w(F) = \log (\rho_F m_F)/\log l_F$ be the \textbf{walk dimension} of the carpet: this is the time-to-space scaling exponent for Brownian motion. For instance, if $\tau(A)=\inf\{t>0: X_t \notin A\}$ is the exit time of a Brownian motion $X_t$ from the set $A$, then $d_w$ is defined through $\mathbb{E}^x[\tau(B(x,r))] \asymp r^{d_w}$, where $\mathbb{E}^x$ is the expectation with respect to the law of Brownian motions started at $x$.

A crucial point to make here is that while manifolds and Euclidean domains have $d_w=2$, fractals have $d_w>2$, as can be seen through \emph{e.g.}~(\ref{eq:resistanceest}). In a nutshell, this means that Brownian motion proceeds "more slowly" on fractals, due to the presence of obstacles at all length scales. One way to see this explicitly is through the \textbf{heat kernel} $p_t(x,y)$, which is the integral kernel of the Markov semigroup $T_t=e^{t\Delta}$, or equivalently, the transition density of Brownian motion. On fractals, regarded as a metric measure space $(F,d,\nu)$, the (short-time) heat kernel obeys sub-Gaussian rather than Gaussian bounds~\cite{BB92}:
\begin{equation}
p_t(x,y) \asymp C t^{-d_h/d_w} \exp\left(-c\left(\frac{d(x,y)^{d_w}}{t}\right)^{\frac{1}{d_w-1}}\right),\quad t \in (0,1), ~x,y \in F.
\label{eq:HKE}
\end{equation}
The leading power-law term is $t^{-d_h/d_w}$, as opposed to the $t^{-d/2}$ on $d$-dimensional Euclidean domains. This observation leads us to introduce the \textbf{spectral dimension} of the carpet $d_s(F) := 2d_h(F)/d_w(F) =2\log m_F /\log(\rho_F m_F)$. Note that for any GSC in $\mathbb{R}^d$, $1 \leq d_s(F) < d_h(F) < d$. 

The spectral dimension is a physically meaningful dimension because it is tied to the macroscopic behavior of Brownian motion. Recall that a Markov process on an unbounded state space is said to be \emph{recurrent} if, with probability one, the process returns to the origin infinitely many times. Otherwise it is said to be \emph{transient}, \emph{i.e.,} with positive probability the process leaves for infinity. Roughly speaking, on a homogeneous space, $d_s=2$ is the dimensional threshold above which Brownian motion is transient, and below which it is recurrent. Brownian motion at $d_s=2$ is typically also recurrent, but its behavior is more subtle than the $d_s<2$ case.

Indeed, for Brownian motion in an unbounded GSC $F_\infty$, Barlow and Bass showed that the process is transient (resp. recurrent) if $d_s(F) > 2$ or $\rho_F < 1$ (resp. $d_s(F) \leq 2$ or $\rho_F \geq 1$)~\cite{BB99}*{Theorem 8.1}. The same dichotomy holds on the pre-carpet $\tilde F$~\cite{BB99}*{Theorem 8.7}. This distinction has important consequences for the thermodynamics of the massive Bose gas, as we will explain later.

To make comparisons, we mention in passing that all post-critically finite (pcf) fractals (for definition see~\cite{KigamiBook, StrichartzBook}), \emph{e.g.} any $d$-dimensional Sierpinski gasket with $d\geq 2$, have $\rho_F >1$ or $d_s<2$, and support recurrent Brownian motion. Meanwhile, as an example of a random fractal, the incipient infinite cluster (IIC) at criticality of bond percolation on $\mathbb{Z}^d$ with $d\geq 19$, or of spread-out percolation with $d > 6$, has $d_s=4/3$, which was proved recently by~\cite{KozmaNachmias}.

\subsection{Sharp estimates of the heat kernel trace} \label{subsec:HKT}

In what follows we call $K(t):={\rm Tr}(e^{t\Delta})=\int_F p_t(x,x) d\nu(x)$ the \emph{heat kernel trace}. Since the Laplacian on GSCs admits an eigenfunction expansion, one can write $K(t) = \sum_{j=0}^\infty e^{-t\lambda_j}$, where $\{\lambda_j\}$ are the eigenvalues of $-\Delta$. Equivalently, $K(t)$ represents the probability that Brownian motion started at any point $x$ returns to the said point $x$ at time $t$.

Using~(\ref{eq:HKE}) we can already deduce that on GSC, $K(t) \asymp t^{-d_h/d_w}$ for small $t$. This estimate is typically enough on simple spaces, but on deterministic self-similar fractals there is logarithmically periodic modulation on top of the power-law dependence in $t$, which is attributed to the discrete scale invariance of the space. Indeed, using renewal theorem type arguments,  Hambly and Kajino separately proved the following short-time asymptotics of the heat kernel trace.

\begin{theorem}[{\cite{HamblySpec,KajinoSpec}}]
Let $F$ be a GSC, $\nu$ be the self-similar measure on $F$, and $\Delta$ be the Laplacian on $(F,\nu)$ associated with either Dirichlet or Neumann condition on $\partial_o F$. Then
\begin{equation}K(t) = t^{-d_h/d_w} \left[G(-\log t) + o(1)\right]\quad \mbox{as}~t\downarrow 0,\label{eq:HKT1}\end{equation}
where $G$ is $\log (\rho_F m_F)$-periodic, and is bounded away from $0$ and $\infty$.
\end{theorem}

\begin{remark}
Let $N(s) := \#\{\lambda <s : \lambda \in {\rm Spec}(-\Delta)\}$ be the eigenvalue counting function (or the integrated density of states) of the Laplacian. Then $K(t) = \int_0^\infty e^{-ts} dN(s)$, \emph{i.e.,} $K(\cdot)$ is the Laplace-Stieltjes transform of $N(\cdot)$. Using the estimate $K(t) \asymp t^{-d_s/2}$ one easily obtains the Weyl asymptotics $N(s) \asymp s^{d_s/2}$ for large $s$. 

In view of (\ref{eq:HKT1}), it is tempting to go a step further and claim that
\begin{equation}N(s) = s^{d_s/2}\left[h(\log s) + o(1)\right]~~~\mbox{as}~~s\to\infty,\label{eq:Weyl}\end{equation}
where $h$ is $\log (\rho_F m_F)$-periodic and is bounded away from $0$ and $\infty$. There is much numerical evidence that~(\ref{eq:Weyl}) holds on two-dimensional and three-dimensional GSCs~\cite{OuterApprox}*{CS, BKS}, as Figure~\ref{fig:SCWeyl} demonstrates. Indeed, if (\ref{eq:Weyl}) is true, then (\ref{eq:HKT1}) follows immediately. However, as pointed out in~\cite{HamblySpec}*{Section 4.2} and \cite{KajinoSpec}*{Section 9}, the inverse Laplace-Stieltjes transform involves sophisticated Tauberian theorems which are very difficult to prove. This explains why we have elevated the role of the heat kernel trace over that of the (integrated) density of states.
\label{re:ECF}
\end{remark}

More recently, Kajino (in preparation) proved a sharper estimate of the heat kernel trace than (\ref{eq:HKT1}) by exploiting the full symmetry of the GSC and its boundaries. 

\begin{theorem}
Let $F$ be a GSC, $\nu$ be the self-similar measure on $F$, and $\Delta$ be the Laplacian on $(F,\nu)$ with Dirichlet conditions on $\partial_o F$. Then there exist continuous, $\log (\rho_F m_F)$-periodic functions $G_k : \mathbb{R} \to \mathbb{R}$ for $k=0,1,\cdots,d$ such that
\begin{equation}
K(t) = \sum_{k=0}^d t^{-d_k/d_w} G_k(-\log t) + \mathcal{O}\left(\exp\left(-ct^{-\frac{1}{d_w-1}}\right)\right)\quad \mbox{as}~t\downarrow 0.
\label{eq:HK}
\end{equation}
Here $d_k := d_h(F \cap \{x_1 = \cdots x_k=0\})$. Moreover $G_0>0$ and $G_1<0$.
\label{prop:HKEst}
\end{theorem}
\begin{remark}
It is believed that the same estimate holds on GSCs with Neumann conditions on $\partial_o F$, but the proof is more elusive. (N. Kajino, personal communication)
\end{remark}

As the $G_k$ are periodic, we can expand them in Fourier series: $G_k(x) = \sum_{p\in \mathbb{Z}} \hat{G}_{k,p} e^{2\pi p i x/\log R}$, where $R := \rho_F m_F = d_w(F) \log l_F$. We should note that very little is known about the functions $G_k$, except for the signs of $G_0$ and $G_1$. There is strong numerical evidence that $G_0$ is nonconstant, with $(\max G_0 -\min G_0)/\hat{G}_{0,0}$ typically on the order of $10^{-2}$ for GSCs in $\mathbb{R}^2$ or $\mathbb{R}^3$~[CS]; but this has not been rigorously established. 

We point out three important features of the estimate~(\ref{eq:HK}) which are crucial to our thermodynamic computations in the sequel:
\begin{enumerate}
\item The polynomial terms all have nonpositive exponents $(-d_k/d_w)_k$, where $d_k$ is the Hausdorff dimension of the codimension-$k$ outer boundary of $F$. In particular, $d_0=d_h(F)$, $d_1=d_h(\partial_o F)$, $d_{d-1}=1$, and $d_{d}=0$. \label{pt:np}
\item Since the $G_k$ are real-valued bounded functions, we have $\hat{G}_{k,p}=\overline{\hat{G}_{k,-p}}$, and $|\hat{G}_{k,p}| \leq \frac{1}{2\pi}\int^{\log R}_0 |G_k(x)|dx \leq C < \infty$ for all $p\in\mathbb{Z}$. \label{pt:bd}
\item The remainder term decays exponentially in $t$, which is better than the power-law decay in the estimate~(\ref{eq:HKT1}).  \label{pt:exp}
\end{enumerate}
The first two features guarantee that the thermodynamic quantities (energy, pressure, etc.) associated with GSCs will be finite and meaningful. The last feature, which is the most important of all, enables us to carry out the Casimir energy calculation unambiguously and without making any \emph{ad hoc} regularization.

Now recall that the heat kernel trace on Riemannian manifolds $M \subset \mathbb{R}^d$  has the short-time asymptotics
\begin{equation}
K(t) = \frac{{\rm Vol}(M)}{(4\pi t)^{d/2}} + \mathcal{O}\left(t^{-(d-1)/2}\right)\quad \mbox{as}~t\downarrow 0.
\label{eq:mfdHK}
\end{equation}
Going back to~(\ref{eq:HK}), it is easy to check that if $F$ is a GSC, then the heat kernel trace for $LF$ ($L>0$) is given by $K(L^{-2} t) = L^{d_s} G_0(-\log(L^{-2}t)) t^{-d_s/2}+ \mbox{lower order terms}$ as $t \downarrow 0$. Comparing this against~(\ref{eq:mfdHK}), we see that $LF$ has an effective spectral volume of $C_2 L^{d_s} (4\pi)^{d_s/2}$, where $C_2 \asymp G_0$. In fact, we can use Ces\`{a}ro averaging to define more precisely the spectral content of a GSC, in the same manner that the Minkowski content of a fractal string is defined~\cite{LapidusBook}. This definition also appeared in~\cite{FractalThermo}.

\begin{definition}
Let $F$ be a GSC. Then we define the \emph{spectral volume} of $LF$ ($L>0$) to be
\begin{equation}
V_s(LF) := L^{d_s}(4\pi)^{d_s/2}\lim_{n\to\infty} \frac{1}{n\log R}\int_{R^{-n}}^{1} G_0(-\log t)\frac{dt}{t} = (4\pi)^{d_s/2} \hat{G}_{0,0} L^{d_s}.
\end{equation}
\end{definition}

We emphasize that the anomalous walk dimension $d_w>2$ of the GSC results in its spectral volume being distinct from its Hausdorff volume. As will be made clear in later sections, all "extensive" thermal observables (\emph{e.g.} energy, particle number) scale with the spectral volume, not the Hausdorff volume.

Moreover we denote $d_{k,p} := 2\left(\frac{d_k}{d_w} + \frac{2 p \pi i}{\log R}\right)$. These are the \emph{complex dimensions} of the GSC in the sense of M. Lapidus~\cite{LapidusBook}, which reflect both the various spectral contents (via the real part of $d_{k,p}$) and the discrete scale invariance of the fractal (via the imaginary part). Using this notation,~(\ref{eq:HK}) can be rewritten as
\begin{equation}
K(t)=\sum_{k=0}^d \sum_{p\in\mathbb{Z}}\hat{G}_{k,p} t^{-d_{k,p}/2}+ \mathcal{O}\left(\exp\left(-ct^{-\frac{1}{d_w-1}}\right)\right)~~~\mbox{as}~~t\downarrow 0.
\label{eq:HK2}
\end{equation}

\section{Spectral zeta function on Sierpinski carpets} \label{sec:speczeta}

The spectral zeta function of a self-adjoint Laplacian $\Delta$ on a bounded domain is given by
$$ \zeta_\Delta(s,\gamma) := {\rm Tr} \frac{1}{(-\Delta+\gamma)^s}.$$
This can be written as a Mellin transform of the heat kernel trace
\begin{equation}
\zeta_\Delta(s,\gamma) = \frac{1}{\Gamma(s)} \int_0^\infty t^s e^{-\gamma t} K(t) \frac{dt}{t}, 
\label{eq:speczetarep}
\end{equation}
whenever the right-hand side is defined. In the case of a GSC with Dirichlet boundary,~(\ref{eq:HK}) implies that the absicssa of convergence for $\zeta_\Delta(\cdot,\gamma)$ is located at ${\rm Re}(s) = d_{0,0}=d_h(F)/d_w$. It is natural to ask whether this function can be extended to the entire complex plane, save for a countable (possibly infinite) number of poles. The following theorem, which was reported in (B. Steinhurst \& A. Teplyaev, arXiv:1011.5485), serves as the linchpin for all subsequent results in this paper.

\begin{theorem}
The spectral zeta function $s\mapsto\zeta_\Delta(s,\gamma)$ of the Laplacian on a GSC $F$, with Dirichlet conditions on $\partial_o F$, admits a meromorphic extension to $\mathbb{C}$.
\label{thm:meroext}
\end{theorem}

\begin{proof}
Without loss of generality we suppose that~(\ref{eq:HK2}) holds for $t\in (0,1)$, while there exist $c_3, c_4>0$ such that $K(t) \leq c_3  e^{-c_4 t}$ for $t>1$. Then $$\zeta_\Delta(s,\gamma) \Gamma(s) = I_1(s,\gamma) + I_2(s,\gamma) +I_3(s,\gamma),$$ 
where
\begin{eqnarray*}
I_1(s,\gamma)&=& \int_0^1 t^s e^{-\gamma t}\sum_{k=0}^d \sum_{p\in\mathbb{Z}}\hat{G}_{k,p} t^{-d_{k,p}/2}\frac{dt}{t},\\
I_2(s,\gamma)&=& \int_0^1 t^s e^{-\gamma t} \mathcal{O}\left(\exp\left(-ct^{-\frac{1}{d_w-1}}\right)\right) \frac{dt}{t},\\
|I_3(s,\gamma)|&\leq&  c_3\left| \int_1^\infty t^s e^{-\gamma t} e^{-c_4 t} \frac{dt}{t}\right|.
\end{eqnarray*}
 
To compute $I_1$ we use the identity
$$\int_0^1 t^s e^{-\gamma t} t^{-p} \frac{dt}{t} = \int_0^1 t^s \sum_{n=0}^\infty \frac{(-1)^n}{n!}\gamma^n t^n t^{-p} \frac{dt}{t} = \sum_{n=0}^\infty \frac{(-1)^n \gamma^n}{n! (s+n-p)}$$
for ${\rm Re}(p)\geq 0$. Using linearity of the polynomial terms we get
$$I_1(s,\gamma)=\sum_{k=0}^d \sum_{p\in\mathbb{Z}}\sum_{n=0}^\infty\frac{(-1)^n \gamma^n \hat{G}_{k,p}}{n! (s+n-d_{k,p}/2)}.$$
In deriving this expression we interchanged the order of $p$-summation and integration, which we justify as follows. For fixed $\gamma \neq 0$, the summand indicates that the simple poles are $d_{k,p}/2-\mathbb{N}_0$. Away from the poles $s\mapsto I_1(s,\gamma)$ is holomorphic, since for each $k=0,1,\cdots,d$ and each $r \in \mathbb{C}$, 
$$\left|\frac{\hat{G}_{k,p}}{r-d_{k,p}/2} + \frac{\hat{G}_{k,-p}}{r-d_{k,-p}/2}\right| \leq C \left| \frac{2 (r-\frac{d_k}{d_w})}{(r-\frac{d_k}{d_w})^2 + (\frac{2p \pi}{\log R})^2} \right|$$
is summable over $p \in \mathbb{N}$. When $\gamma=0$, we find
$$I_1(s,0) = \sum_{k=0}^d \sum_{p\in\mathbb{Z}} \frac{\hat{G}_{k,p}}{s-d_{k,p}/2}.$$

Next we let $\alpha := 1/(d_w-1)$, and make the bound
$$
\left|\frac{I_2(s,\gamma)}{\Gamma(s)}\right| \leq C \left|\frac{1}{\Gamma(s)}\int_0^1 t^s e^{-\gamma t} e^{-ct^{-\alpha}} \frac{dt}{t} \right| \leq C e^{-c} \left|\frac{\gamma^{-s}}{\Gamma(s)}\boldsymbol{\gamma}(s,\gamma)\right|,
$$
where $\boldsymbol{\gamma}(s,x)$ is the lower incomplete gamma function. Since
$$ \boldsymbol{\gamma}^*(s,\gamma) := \frac{\gamma^{-s}}{\Gamma(s)}\boldsymbol{\gamma}(s,\gamma) = e^{-\gamma}\sum_{m=0}^\infty \frac{\gamma^m}{\Gamma(s+m+1)}$$
is entire in both $s$ and $\gamma$, $I_2(s,\gamma)/\Gamma(s)$ contributes an entire-in-$s$ component to $\zeta_\Delta(s,\gamma)$.

For $I_3$, it suffices to know that for ${\rm Re}(\gamma) > -c_4$, 
$$|I_3(s,\gamma)| \leq c_3 \left|(\gamma + c_4)^{-s} \int_{\gamma+c_4}^\infty u^s e^{-u} \frac{du}{u}\right| =c_3 \left| (\gamma + c_4)^{-s} \Gamma(s,\gamma+c_4)\right|$$
is finite, thanks to the holomorphicity of the upper incomplete gamma function $s\mapsto \Gamma(s,x)$ when $x\neq 0$.

So finally we obtain, for ${\rm Re}(\gamma)>-c_4$, the Mittag-Leffler decomposition
\begin{equation}
\zeta_\Delta(s,\gamma) =  \frac{1}{\Gamma(s)}\sum_{k=0}^d \sum_{p\in\mathbb{Z}}\sum_{n=0}^\infty \frac{(-1)^n\gamma^n  \hat{G}_{k,p}}{n!\left(s+n-d_{k,p}/2\right)} + (\mbox{entire function in $s$}).
\label{eq:ML}
\end{equation}
This proves the theorem.
\qed \end{proof}

Now we can enumerate all the possible poles and residues of the spectral zeta function on a GSC.

\begin{corollary}
For a GSC $F$ with Dirichlet condition on $\partial_o F$:
\begin{enumerate}
\item The poles of $\zeta_\Delta(\cdot,\gamma)$, where ${\rm Re}(\gamma)\geq -c_4$, are contained in the set
$$ \bigcup_{k=0}^d \bigcup_{p\in \mathbb{Z}} \bigcup_{n=0}^\infty \left\{-n+\frac{d_{k,p}}{2}\right\},$$
with corresponding residues
$$ {\rm Res}\left(\zeta_\Delta(\cdot,\gamma),-n+\frac{d_{k,p}}{2}\right) = \frac{(-1)^n \gamma^n \hat{G}_{k,p}}{n! \Gamma\left(-n+d_{k,p}/2\right)}.$$

\item The poles of $\zeta_\Delta(\cdot,0)$ are contained in the set
$\displaystyle\bigcup_{k=0}^d \bigcup_{p\in \mathbb{Z}} \left\{\frac{d_{k,p}}{2}\right\}$,
with corresponding residues
$$ {\rm Res}\left(\zeta_\Delta(\cdot,0),\frac{d_{k,p}}{2}\right) = \frac{\hat{G}_{k,p}}{\Gamma\left(d_{k,p}/2\right)}.$$
In particular, $\zeta_\Delta(s,0)$ is analytic for ${\rm Re}(s)<0$.

\item $\zeta_\Delta(s,0)=0$ for all $s\in-\mathbb{N}$.
\end{enumerate}
\label{cor:polesresidues}
\end{corollary}

\begin{proof}
The poles and residues can be read off from~(\ref{eq:ML}), so the first two claims are immediate. Note that $d_{d,0}/2-\mathbb{N} = -\mathbb{N}$ are poles of $\zeta_\Delta(\cdot,\gamma)$ with zero residue because $1/\Gamma(s)=0$ for all $s \in -\mathbb{N}$. The last claim follows from the observation that  when ${\rm Re}(s)<0$, $\zeta_\Delta(s,0)$ is a product of $1/\Gamma(s)$ and a holomorphic function. 
\qed \end{proof}

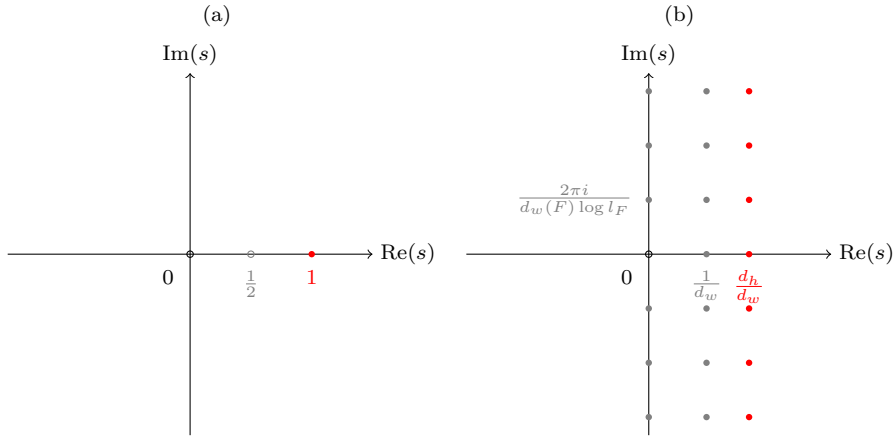
\begin{figure}
\centering
\subfigure[~]{
\begin{tikzpicture}[scale=0.8]
    \draw[axis_line, ->] (-3,0) -- (3.0,0) node[textlabel,anchor=west]{${\rm Re}(s)$};
    \draw[axis_line, ->] (0,-3.0) -- (0,3.0) node[textlabel,anchor=south]{${\rm Im}(s)$};
    \coordinate (s2) at (2,0);
    \coordinate (s1) at (1,0);
    \coordinate (s0) at (0,0);
    \fill[red] (s2) circle (1.5pt) node[label=below:\small{$1$}] {};
    \draw[gray] (s1) circle (1.5pt) node[label=below:\small{$\frac{1}{2}$}] {};
    \draw[black] (s0) circle (1.5pt) node[label=below left:\small{$0$}] {};
\end{tikzpicture}
\label{subfig:zeta2d}    
}
\subfigure[]{
\begin{tikzpicture}[scale=0.8]
    \draw[axis_line, ->] (-3,0) -- (3.0,0) node[textlabel,anchor=west]{${\rm Re}(s)$};
    \draw[axis_line, ->] (0,-3.0) -- (0,3.0) node[textlabel,anchor=south]{${\rm Im}(s)$};
\coordinate (c0) at (1.65,0);
\coordinate (c1) at (1.65,0.9);
\coordinate (c2) at (1.65,1.8);
\coordinate (c3) at (1.65,2.7);
\coordinate (c1-) at (1.65,-0.9);
\coordinate (c2-) at (1.65,-1.8);
\coordinate (c3-) at (1.65,-2.7);
\coordinate (b0) at (0.95,0.0);
\coordinate (b1) at (0.95,0.9);
\coordinate (b2) at (0.95,1.8);
\coordinate (b3) at (0.95,2.7);
\coordinate (b1-) at (0.95,-0.9);
\coordinate (b2-) at (0.95,-1.8);
\coordinate (b3-) at (0.95,-2.7);
\coordinate (a0) at (0,0);
\coordinate (a1) at (0,0.9);
\coordinate (a2) at (0,1.8);
\coordinate (a3) at (0,2.7);
\coordinate (a1-) at (0,-0.9);
\coordinate (a2-) at (0,-1.8);
\coordinate (a3-) at (0,-2.7);
\fill[red] (c0) circle (1.5pt) node[label=below:\small{$\frac{d_h}{d_w}$}] {};
\fill[red] (c1) circle (1.5pt);
\fill[red] (c2) circle (1.5pt);
\fill[red] (c3) circle (1.5pt);
\fill[red] (c1-) circle (1.5pt);
\fill[red] (c2-) circle (1.5pt);
\fill[red] (c3-) circle (1.5pt);
\fill[gray] (b0) circle (1.5pt) node[label=below:\small{$\frac{1}{d_w}$}] {};
\fill[gray] (b1) circle (1.5pt);
\fill[gray] (b2) circle (1.5pt);
\fill[gray] (b3) circle (1.5pt);
\fill[gray] (b1-) circle (1.5pt);
\fill[gray] (b2-) circle (1.5pt);
\fill[gray] (b3-) circle (1.5pt);

\draw[black] (a0) circle (1.5pt) node[label=below left:\small{$0$}] {};
\fill[gray] (a1) circle (1.5pt) node[label=left:\small{$\frac{2\pi i}{d_w(F) \log l_F}$}] {};
\fill[gray] (a2) circle (1.5pt);
\fill[gray] (a3) circle (1.5pt);
\fill[gray] (a1-) circle (1.5pt);
\fill[gray] (a2-) circle (1.5pt);
\fill[gray] (a3-) circle (1.5pt);

\end{tikzpicture}
\label{subfig:zetaGSC}
}

\caption{The singularities of $s\mapsto\zeta_\Delta(s,0)$ associated to~\subref{subfig:zeta2d} a 2-dimensional parallelpiped, using the Epstein zeta function;~\subref{subfig:zetaGSC} a 2-dimensional GSC, using Corollary~\ref{cor:polesresidues}. A hollow circle indicates that the singularity is either a removable singularity, or a pole with zero residue.}
\label{fig:pole}
\end{figure}

Fig.~\ref{fig:pole} depicts the singularities of $\zeta_\Delta(\cdot,0)$ for a typical GSC versus that for a parallelpiped. In both domains, the poles with the largest real part represent the spectral volume dimension. What distinguishes a GSC from an Euclidean domain is the "tower" of poles, or \emph{complex dimensions}, along the imaginary direction at each ${\rm Re}(s) = d_{k,0}$, with even spacing of $2\pi i/(d_w\log l_F)$. Thus the discrete scale invariance of the GSC manifests itself in two ways: it produces log-periodic modulations of the heat kernel trace and, via a Mellin transform, creates towers of complex-valued poles in the spectral zeta function.  

\section{Review of quantum statistical mechanics}\label{sec:QSM}

Here we briefly review notions of quantum statistical mechanics that are needed for the rest of the paper. Much of this is known to physicists, but we use the language of functional analysis to make notations precise. For a complete treatment please consult, e.g.,~\cite{BratelliRobinson}*{Chapter 5}.

\subsection{Gibbs state and partition function}

The Bose (resp. Fermi) gas is a system of many quantum particles satisfying an appropriate rule of quantum statistics. Let us denote the Hilbert space for a single quantum particle confined to domain $F \subset \mathbb{R}^d$ by $\mathcal{H}_1 := L^2 (F)$: the state of the particle is described by a wavefunction $\psi \in \mathcal{H}_1$. For a system of $n$ such identical bosons (resp. fermions), the wavefunction $\psi_n(x_1,\cdots, x_n) \in L^2(F^n)$ is (anti)symmetric under particle exchange: $\psi_n(\cdots, x_i,\cdots, x_j,\cdots) = \pm \psi_n(\cdots, x_j, \cdots, x_i, \cdots)$. Therefore the appropriate $n$-body Hilbert space, denoted by $\mathcal{H}_n$, is the (anti)symmetric subspace of $\mathcal{H}_1^{\otimes n}$. By default we set $\mathcal{H}_0 = \mathbb{C}$.

To describe the equilibrium thermodynamics of the quantum many-body system, it is convenient to adopt the \emph{grand canonical ensemble}, where the particle number is not kept fixed. Mathematically, we take the \emph{Fock space} $\mathcal{F} = \overline{\bigoplus_{n=0}^\infty \mathcal{H}_n}$ as the underlying Hilbert space. For each one-body self-adjoint operator $W:\mathcal{H}_1 \to \mathcal{H}_1$, define $W_n : \mathcal{H}_n \to \mathcal{H}_n$ by
$$W_n P_{\pm} (\psi_1 \otimes \cdots \otimes \psi_n) =\left\{ \begin{array}{ll}0,&\mbox{if $n=0$} \\ \displaystyle \sum_{j=1}^n P_{\pm} \left(\psi_1 \otimes \cdots \otimes (W\psi_j) \otimes \cdots \otimes \psi_n\right),& \mbox{if $n \geq 1$}\end{array}\right.$$
for all $\psi_j \in \mathcal{H}_1$, where $P_{\pm}$ is the (anti)symmetric permutation on the components of the tensor product. Then we call $d\Gamma(W) := \overline{\bigoplus_{n=0}^\infty W_n}$ the \emph{second quantization} of $W$ on $\mathcal{F}$. An important example is where $W$ is the identity operator ${\bf 1}$ on $\mathcal{H}_1$: then ${\bf 1}_n$ is $n$ times the identity operator on $\mathcal{H}_n$, so its second quantization reads $d\Gamma({\bf 1}) = \oplus_{n=0}^\infty n =: {\bf N}$, also known as the \emph{number operator} on $\mathcal{F}$.

For our purposes it suffices to consider a one-body Hamiltonian operator $H$ which is bounded below. Let $\beta \in (0,\infty]$ be the inverse temperature, $\mu \in (-\infty,\inf{\rm Spec}(H)]$ be the chemical potential, and ${\bf H}:=d\Gamma(H)$. Consequently $\rho_{\beta,\mu}:=e^{-\beta d\Gamma(H-\mu {\bf 1})}$ is a trace-class operator on $\mathcal{F}$, and its trace
$$\Xi_{\beta,\mu} := {\rm Tr}_{\mathcal{F}} \rho_{\beta,\mu}={\rm Tr}_\mathcal{F} e^{-\beta({\bf H}-\mu {\bf N})}$$
is called the \emph{grand canonical partition function} of a thermal system with Hamiltonian $H$. The \emph{free energy} of the system is
$$ F_{\beta,\mu}=-\beta^{-1} \log\Xi_{\beta,\mu}.$$
The \emph{Gibbs state} $\omega_{\beta,\mu}$ at $(\beta,\mu)$ is a linear functional over the $C^*$-algebra on $\mathcal{F}$ such that
$$\omega_{\beta,\mu}({\bf A}) = \Xi_{\beta,\mu}^{-1}{\rm Tr}_{\mathcal{F}}\left({\bf A}\rho_{\beta,\mu}\right).$$
Physically, if $A$ belongs to the (*-)algebra of observables on $\mathcal{H}_1$, and ${\bf A}=d\Gamma(A)$, then $\omega_{\beta,\mu}({\bf A})$ represents the expectation value of $A$ in the grand canonical Gibbs state. To experts of operator algebra, we remark that the Gibbs state is a $(\tau,\beta)$-Kubo-Martin-Schwinger (KMS) state~\cite{BratelliRobinson}*{Definition 5.3.1}: namely, for all ${\bf A}, {\bf B}$ in the algebra of observables on $\mathcal{F}$,
$$\omega_{\beta,\mu}({\bf A}\tau_{i\beta}({\bf B})) = \omega_{\beta,\mu}({\bf B}{\bf A}),\quad~\mbox{where}~ \tau_t({\bf O}) := e^{it({\bf H}-\mu {\bf N})} {\bf O} e^{-it({\bf H}-\mu {\bf N})}.$$

Two observables we will be interested in the sequel are the expected number of bosons 
$$ \omega_{\beta,\mu}({\bf N}) =  \beta^{-1} \frac{1}{\Xi_{\beta,\mu}} \frac{\partial \Xi_{\beta,\mu}}{\partial \mu}=\beta^{-1} \frac{\partial}{\partial \mu} \log \Xi_{\beta,\mu},$$
and the expected energy
$$ \omega_{\beta,\mu}({\bf H}) = -\frac{1}{\Xi_{\beta,\mu}}\frac{\partial \Xi_{\beta,\mu}}{\partial \beta} = -\frac{\partial}{\partial \beta} \log\Xi_{\beta,\mu}.$$

For a non-interacting Bose gas the expression for the partition function can be reduced to a trace over the one-body Hilbert space $\mathcal{H}_1$. Assume that $H$ has pure point spectrum with increasing eigenvalues $0 \leq E_0 < E_1 \leq \cdots \leq E_i \leq \cdots \uparrow \infty$. Then
$$\Xi_{\beta,\mu}= \prod_i \sum_{n_i=0}^\infty e^{-n_i\beta (E_i-\mu)} = \prod_i \frac{1}{1-e^{-\beta(E_i-\mu)}}.$$
Therefore
\begin{eqnarray*}
\log\Xi_{\beta,\mu} &=& -{\rm Tr}_{\mathcal{H}_1} \log(1-e^{-\beta(H-\mu)}),\\
\omega_{\beta,\mu}({\bf N}) &=& {\rm Tr}_{\mathcal{H}_1} \frac{1}{e^{\beta(H-\mu)}-1} ,\\
\omega_{\beta,\mu}({\bf H}) &=& {\rm Tr}_{\mathcal{H}_1} \frac{H e^{\beta(H-\mu)}}{e^{\beta(H-\mu)}-1}.
\end{eqnarray*}
The following fact is useful and will be invoked several times: $\mu\mapsto\omega_{\beta,\mu}({\bf N})$ is a convex, monotone increasing function which maps $(-\infty,\inf{\rm Spec}(H)]$ to $(0,\infty]$.

\subsection{Thermodynamic limit}

Since we will deal with the delicate topic of phase transitions on fractals, it is important to discuss how we take the thermodynamic limit. As is customary, we exhaust an unbounded space $S_\infty$ by an increasing family of finite subsets $\{S_n\}_n$ with suitable boundary conditions. The thermodynamic limit of an "intensive" thermal observable (\emph{e.g.} pressure, particle density) on $S_\infty$ is then understood as the limit of a sequence of the said observables on $S_n$, subject to proper constraints.

For the unbounded carpet $F_\infty$, we choose the natural exhaustion $\{\Lambda_n\}_n :=\{ l_F^n F\}_n$. This involves scaling up the size of the system by $l_F$ each time, which in effect scales down the Laplacian by $l_F^{-2}$. To make this length dependence explicit, we will denote by $\underline\Delta$ the dimensionful Laplacian on a bounded domain not necessarily of unit characteristic length. If the characteristic length is $L$, then we write $\underline\Delta_L$, which is $L^{-2}$ times the dimensionless Laplacian $\Delta$ (on the domain of unit length). In the same manner we will add a subscript $L$ to any quantity which depends explicitly on the domain size $L$.

We will show shortly that, due to the log-periodic modulation inherited from the heat kernel trace, the thermodynamic limit along the natural exhaustion does not exist in the literal sense, but rather in the weaker limsup/liminf sense. Of course our choice of exhaustion is by no means unique. There might be other choices of exhaustion for which a stronger limit is attained, but we suspect that those would marginally improve the results to be presented below.

\section{Massive Bose gas in Sierpinski carpets} \label{sec:atoms}

In this section we focus on the Bose gas with $H=-\Delta$ and $\mu \leq \inf{\rm Spec}(H)$. This Hamiltonian captures the behavior of a gas of non-relativistic, non-interacting, and charge-neutral atoms. The partition function of the massive Bose gas in a general bounded domain is given by
\begin{equation}\log\Xi_{\beta,\mu}= -{\rm Tr}_{\mathcal{H}_1}\log\left(1-e^{-\beta(-\underline{\Delta}-\mu)}\right).\label{eq:BosePartitionFcn}\end{equation}
Its connection to the spectral zeta function is the following.
\begin{proposition}
The grand canonical partition function $\Xi_{\beta,\mu}$ for an ideal massive Bose gas in a bounded domain $F$ is given by
\begin{equation}
\log\Xi_{\beta,\mu}=\frac{1}{2\pi i}\int^{\sigma+i\infty}_{\sigma-i\infty}\beta^{-t} \Gamma(t) \zeta(t+1) \zeta_{\underline\Delta}\left(t,-\mu\right)dt,\quad \sigma > \frac{d_s}{2}.
\label{eq:Bosepart}
\end{equation}
\label{prop:Bosepart}
\end{proposition}
\begin{proof}
Using the Taylor expansion for $\log(1-x)$ and functional calculus, we expand the RHS of~(\ref{eq:BosePartitionFcn}) in a power series: 
$$
\log\Xi_{\beta,\mu} = {\rm Tr}_{\mathcal{H}_1} \sum_{n=1}^\infty \frac{1}{n} e^{-n\beta (-\underline\Delta-\mu)}.
$$
By the identity
$$ e^{-a} = \frac{1}{2\pi i} \int_{\sigma-i\infty}^{\sigma+i\infty} a^{-t}\Gamma(t) dt, \quad \sigma>0,
$$
we have
$$
\log\Xi_{\beta,\mu}={\rm Tr}_{\mathcal{H}_1}\frac{1}{2\pi i}\sum_{n=1}^\infty \frac{1}{n} \int_{\sigma-i\infty}^{\sigma+i\infty} (n\beta)^{-t} \Gamma(t) (-\underline{\Delta}-\mu)^{-t} dt.
$$
It remains to interchange the order of integration and summations. To do so we need to pick $\sigma$ such that both $\sum_{n=1}^\infty n^{-(\sigma+ic+1)}$ and ${\rm Tr}_{\mathcal{H}_1} (-\underline\Delta-\mu)^{-(\sigma+ic)}$ converge for any $c\in \mathbb{R}$. This happens when $\sigma >  0 \vee \frac{d_s}{2} = \frac{d_s}{2}$.
\qed \end{proof}

\subsection{Bose-Einstein condensation in the unbounded carpet}

For the massive Bose gas we are interested in the density of bosons, that is, the number of bosons per unit spectral volume:
\begin{equation}
\rho_L(\beta,\mu) := \frac{\omega_{L,\beta,\mu}({\bf N})}{V_s(L)} =\frac{1}{V_s(L)} {\rm Tr}_{\mathcal{H}_1} \frac{1}{e^{\beta(-\underline\Delta_L-\mu)}-1} =\frac{1}{\beta V_s(L)}\frac{\partial}{\partial \mu} \log \Xi_{L,\beta,\mu}.
\label{eq:defdensity}
\end{equation}

\begin{lemma}
Let $F$ be a GSC. Then the density of ideal massive Bose gas in $LF$ at $(\beta,\mu)$ is
\begin{equation}
\rho_L(\beta,\mu) = \frac{1}{(4\pi)^{\frac{d_s}{2}} \hat{G}_{0,0}}\left(\sum_{k=0}^d \sum_{p\in\mathbb{Z}}\right)' \frac{L^{d_{k,p}-d_s}}{\beta^{d_{k,p}/2}} \hat{G}_{k,p} \sum_{m=0}^\infty\frac{(\beta\mu)^m}{m!}\zeta\left(-m+\frac{d_{k,p}}{2}\right),
\label{eq:rhoL}
\end{equation}
where $\left(\sum_{k=0}^d \sum_{p\in\mathbb{Z}}\right)'$ means the double sum excluding the term $(k,p)=(d,0)$.
\label{lem:rhoc}
\end{lemma}

\begin{proof}
We evaluate the contour integral in~(\ref{eq:Bosepart}) via a residue calculation. By Corollary~\ref{cor:polesresidues}, $\zeta_\Delta(\cdot, \gamma)$ has poles $\bigcup_{n\in \mathbb{N}_0} (d_{k,p}/2-n)$ with residues $\frac{(-1)^n \gamma^n \hat{G}_{k,p}}{n! \Gamma(-n+d_{k,p}/2)}$. Upon excluding the poles with zero residue, $d_{d,0}/2-\mathbb{N}_0$, one finds
$$
\log\Xi_{L,\beta,\mu} = \left(\sum_{k=0}^d \sum_{p\in\mathbb{Z}}\right)'  \left(\frac{L^2}{\beta}\right)^{d_{k,p}/2}\sum_{m=0}^\infty \frac{(\beta\mu)^m \hat{G}_{k,p}}{m!}\zeta\left(\frac{d_{k,p}}{2}-m+1\right).
$$
The lemma follows from taking the $\mu$-derivative of $\log \Xi_{L,\beta,\mu}$ and then dividing by the spectral volume $V_s(L) = (4\pi)^{\frac{d_s}{2}}\hat{G}_{0,0}L^{d_s}$.
\qed \end{proof}

We are in a position now to take the thermodynamic limit of the particle density along $\Lambda_n \nearrow F_\infty$. To streamline the arguments, we introduce the \emph{fugacity} $z=e^{\beta\mu}$, and state all quantities in terms of $z$ rather than $\mu$. We write $\rho_{\Lambda_n}(\beta,z)$ for the particle density in $\Lambda_n$ at $(\beta,z)$.

\begin{proposition}
Let $\displaystyle\overline{\rho}_c(\beta) = \limsup_{n\to\infty}\rho_{\Lambda_n}(\beta,1)$ and $\displaystyle\underline\rho_c(\beta) = \liminf_{n\to\infty} \rho_{\Lambda_n}(\beta,1)$ be, respectively, the upper and lower critical density of the massive Bose gas in $F_\infty$. Then $\underline\rho_c(\beta)=\overline\rho_c(\beta)=\infty$ if $d_s\leq 2$, while $0 < \underline\rho_c(\beta) \leq \overline\rho_c(\beta) <\infty$ if $d_s>2$. In the latter case,
\begin{equation}
\overline\rho_c(\beta) = \frac{1}{(4\pi \beta)^{\frac{d_s}{2}}}\frac{\max G_0}{\hat{G}_{0,0}} \zeta\left(\frac{d_s}{2}\right),~~\underline\rho_c(\beta) = \frac{1}{(4\pi \beta)^{\frac{d_s}{2}}}\frac{\min G_0}{\hat{G}_{0,0}} \zeta\left(\frac{d_s}{2}\right).
\label{eq:density}
\end{equation}
\label{prop:uppdensity}
\end{proposition}

\begin{proof}
First we note that as $L\to\infty$, only the volume terms ($k=0$) dominate:
$$
 \rho_L(\beta,\mu) = \frac{1}{(4\pi \beta)^{\frac{d_s}{2}}\hat{G}_{0,0}}\sum_{p\in\mathbb{Z}} \left(\frac{L^2}{\beta}\right)^{\frac{2\pi ip}{\log R}}\hat{G}_{0,p} \sum_{m=0}^\infty \frac{(\beta\mu)^m}{m!}\zeta\left(-m+\frac{d_{0,p}}{2}\right) + o(1).
$$
By expanding the Riemann zeta function in a series and using some manipulation, we arrive at a more transparent expression
$$
\rho_L(\beta,z)= \frac{1}{(4\pi\beta)^{\frac{d_s}{2}} \hat{G}_{0,0}}\sum_{m=1}^\infty z^m G_0\left(-\log\left(\frac{m\beta}{L^2}\right)\right) m^{-\frac{d_s}{2}} + o(1).
\label{eq:rhoLsimp}
$$
Replacing $L$ with $l_F^n$ we conclude that as $n\to\infty$,
$$
\rho_{\Lambda_n}(\beta,z)= \frac{1}{(4\pi\beta)^{\frac{d_s}{2}} \hat{G}_{0,0}}\sum_{m=1}^\infty z^m G_0\left(-\log\left(\frac{m\beta}{(l_F)^{2n}}\right)\right) m^{-\frac{d_s}{2}} + o(1),
$$
from which the proposition follows. Finiteness comes from the convergence of the Riemann zeta function $\zeta(d_s/2)=\sum_{m=1}^\infty m^{-d_s/2}$. Note that we do not claim that $\underline\rho_c(\beta) =\overline\rho_c(\beta)$ since $G_0$ is probably nonconstant.
\qed \end{proof}

The consequence of a finite upper critical density is that any excess Bose gas must occupy the lowest eigenfunction of the Hamiltonian, a phenomenon known as \emph{Bose-Einstein condensation (BEC)}. This is the content of the next theorem. We will state the result for a fixed-density Bose gas, \emph{i.e.,} fix the particle density $\rho_{\rm tot}$ along $\Lambda_n \nearrow F_\infty$. This of course requires adjusting the fugacity $z_n$ for each $n$. Moreover, we denote by $E_k(L)$ and $E_k(\Lambda_n)$ the $(k+1)$th eigenvalue of the Laplacian on, respectively, $LF$ and $\Lambda_n$.

\begin{theorem}
Assume $d_s>2$. For each $\rho_{\rm tot}>0$, let $z_n$ be the unique root of $\rho_{\Lambda_n}(\beta,z_n) = \rho_{\rm tot}$. 
\begin{enumerate}[label={(\roman*)},nolistsep]
\item If $\rho_{\rm tot} \leq \overline\rho_c(\beta)$, and $\overline z$ is the root of $\displaystyle\limsup_{n\to\infty}\rho_{\Lambda_n}(\beta,\overline z)=\rho_{\rm tot}$, then $\displaystyle\liminf_{n\to\infty} z_n=\overline z$. \label{BECthm1}
\item If $\rho_{\rm tot} > \overline\rho_c(\beta)$, then $\displaystyle\lim_{n\to\infty}z_n =1$. Moreover if we denote the boson occupation density in the ground state, or the \textbf{condensate density}, by
$$\rho^0_{\Lambda_n}(\beta,z) := \frac{1}{V_s(\Lambda_n)}\frac{1}{z^{-1}e^{\beta E_0(\Lambda_n)}-1},$$
then $\displaystyle \lim_{n\to\infty}\left[\rho^0_{\Lambda_n}(\beta,z_n) + \rho_{\Lambda_n}(\beta,1)\right]=\rho_{\rm tot}$.  
\label{BECthm2}
\end{enumerate}
\label{thm:BEC}
\end{theorem}

The proof, which mirrors that of~\cite{BratelliRobinson}*{Theorem 5.2.30} with some modifications, relies upon extensive use of the convexity of $z\mapsto \rho_L(\beta,z)$ and the lemma below.

\begin{lemma}
For any $z_1 > z_2$,
\begin{equation}
\frac{\rho_L(\beta,z_2)}{z_2} \leq \frac{\rho_L(\beta,z_1)-\rho_L(\beta,z_2)}{z_1-z_2} \leq \frac{\rho_L(\beta,z_1)}{z_1(1-z_1 e^{-\beta E_0(L)})}.
\label{eq:convexineq}
\end{equation}
Analogously, if we denote the boson occupation density above the $(m+1)$th eigenfunction by
\begin{equation}
\rho_L^{m+}(\beta,z) := \frac{1}{V_s(L)}\sum_{k>m} \frac{1}{z_L^{-1} e^{\beta E_k(L)}-1},
\end{equation}
then for any $z_1>z_2$,
\begin{equation}
\frac{\rho_L^{m+}(\beta,z_2)}{z_2} \leq \frac{\rho_L^{m+}(\beta,z_1)-\rho_L^{m+}(\beta,z_2)}{z_1-z_2} \leq \frac{\rho_L^{m+}(\beta,z_1)}{z_1(1-z_1 e^{-\beta E_m(L)})}.
\label{eq:convexineq2}
\end{equation}
\label{lem:convexineq}
\end{lemma}
\begin{proof}
Using the convexity of $z\mapsto \rho_L(\beta,z)$, we have for any $z_1>z_2$,
\begin{equation}
\frac{\partial \rho_L}{\partial z}(\beta,z_2) \leq \frac{\rho_L(\beta,z_1)-\rho_L(\beta,z_2)}{z_1-z_2} \leq \frac{\partial \rho_L}{\partial z}(\beta,z_1),
\label{eq:convexity}
\end{equation}
where the derivative is easily computed:
\begin{eqnarray}
\nonumber \frac{\partial \rho_L}{\partial z}(\beta,z) &=& \frac{1}{V_s(L)}{\rm Tr}_{\mathcal{H}_1} \frac{e^{\beta \underline\Delta_L}}{\left(1-z e^{\beta \underline\Delta_L}\right)^2} \\
& =& \frac{1}{V_s(L)}
{\rm Tr}_{\mathcal{H}_1} \left[\frac{1}{z^{-1} e^{-\beta \underline\Delta_L}-1} \cdot \frac{z^{-1}}{1-ze^{\beta \underline\Delta_L}}\right].
\label{eq:derivrho}
\end{eqnarray}
Since $-\underline\Delta_L \geq E_0(L)$ and $z \leq e^{-\beta \underline\Delta_L}$ by assumption, we can make the bound
\begin{equation}
z^{-1}\leq\frac{z^{-1}}{1-ze^{\beta\underline\Delta_L}} \leq \frac{z^{-1}}{1-ze^{-\beta E_0(L)}}.
\label{eq:bound}
\end{equation}
Applying~(\ref{eq:derivrho}) and~(\ref{eq:bound}) to~(\ref{eq:convexity}) yields~(\ref{eq:convexineq}). The same argument carries over to $\rho_L^{m+}$ by changing $E_0(L)$ to $E_m(L)$ in the upper bound.
\qed \end{proof}

\begin{proof}[of Theorem~\ref{thm:BEC}]
Since $G_0>0$ and $G_1<0$, we have for all sufficiently large $n$,
\begin{eqnarray*}
\rho_{\Lambda_n}(\beta,z) &\leq& \frac{1}{(4\pi\beta)^{\frac{d_s}{2}} \hat{G}_{0,0}} \sum_{m=1}^\infty z^m G_0 \left(-\log\left(\frac{m\beta}{(l_F)^{2n}}\right)\right)m^{-\frac{d_s}{2}} \label{eq:ineq} \\ 
\nonumber & \leq& \frac{1}{(4\pi\beta)^{\frac{d_s}{2}} \hat{G}_{0,0}} \sum_{m=1}^\infty z^m m^{-\frac{d_s}{2}} \max(G_0)= \limsup_{n\to\infty} \rho_{\Lambda_n}(\beta,z).
\end{eqnarray*}

Now we prove~\ref{BECthm1}. By the preceding inequality~(\ref{eq:ineq}) and the convexity of $\rho_L(\beta,\cdot)$, we deduce that $z_n \geq \overline z$ for all sufficiently large $n$. Then together with the first inequality in~(\ref{eq:convexineq}) of Lemma~\ref{lem:convexineq} with $z_1=z_n$ and $z_2 =\overline z$, we get
$$
0\leq z_n - \overline z \leq \overline z\left(\frac{\rho_{\Lambda_n}(\beta,z_n)}{\rho_{\Lambda_n}(\beta,\overline z)}-1\right) = \overline z \left(\frac{\rho_{\rm tot}}{\rho_{\Lambda_n}(\beta,\overline z)}-1\right).
$$
Taking the $\liminf$ of this inequality yields the result.

To prove~\ref{BECthm2}, suppose $\rho_{\rm tot} > \overline\rho_c(\beta)$. If $z_n  \leq 1$ for some large  $n$, then $\displaystyle \overline\rho_c(\beta) < \rho_{\Lambda_n}(\beta,z_n) \leq \limsup_{n\to\infty} \rho_{\Lambda_n}(\beta,z_n) \leq \limsup_{n\to\infty} \rho_{\Lambda_n}(\beta,1) =: \overline\rho_c(\beta)$, which is a contradiction. Thus $z_n >1$. But since $z_n \leq e^{\beta E_0(\Lambda_n)}$ we have $\displaystyle \lim_{n\to\infty} z_n =1$. Moreover $\rho_{\Lambda_n}^{m+}(\beta,z_L) > \rho_{\Lambda_n}^{m+}(\beta,1)$ for any $m\in\mathbb{N}_0$, so in particular
$$
\rho_{\Lambda_n}^{0+}(\beta,z_n) > \rho_{\Lambda_n}^{0+}(\beta,1) \geq \rho_{\Lambda_n}(\beta,1) - \frac{1}{V_s(\Lambda_n) \beta E_0(\Lambda_n)},
$$
and
\begin{eqnarray*}
\rho^0_{\Lambda_n}(\beta,z_n) + \rho_{\Lambda_n}(\beta,1)&\leq& \rho^0_{\Lambda_n}(\beta,z_n) + \rho^{0+}_{\Lambda_n}(\beta,z_n) + \frac{1}{V_s(\Lambda_n)\beta E_0(\Lambda_n)} \\
 &= & \rho_{\rm tot} +\frac{1}{V_s(\Lambda_n)\beta E_0(\Lambda_n)}.
\end{eqnarray*}
Noting that $[V_s(\Lambda_n)E_0(\Lambda_n)]^{-1} = Cn^{2-d_s} \to 0$ as $n\to\infty$, we find 
$$ \liminf_{n\to\infty} \left[\rho^0_{\Lambda_n}(\beta,z_n) + \rho_{\Lambda_n}(\beta,1)\right] \leq \rho_{\rm tot}.$$

It remains to show that $\displaystyle \limsup_{n\to\infty} \left[\rho^0_{\Lambda_n}(\beta,z_n) + \rho_{\Lambda_n}(\beta,1)\right] \geq \rho_{\rm tot}$. Here we utilize the fact that for any $m \in \mathbb{N}$ and $n\in\mathbb{N}_0$, the occupation density in the $(m+1)$th eigenfunction of the Laplacian on $\Lambda_n$
\begin{eqnarray*}
\frac{1}{V_s(\Lambda_n)} \frac{1}{z_n^{-1}e^{\beta E_m(\Lambda_n)}-1} &\leq& \frac{1}{V_s(\Lambda_n)} \frac{1}{e^{\beta(E_m(\Lambda_n) - E_0(\Lambda_n))}-1}\\ &\leq& \frac{1}{V_s(\Lambda_n)}\frac{1}{\beta(E_m(\Lambda_n)-E_0(\Lambda_n))}.
\end{eqnarray*}
The right-hand side scales with $n^{2-d_s}$ and hence tends to $0$ as $n\to\infty$, independently of $m$. This implies that
$$\rho_{\rm tot} = \lim_{n\to\infty} \left[\rho^0_{\Lambda_n}(\beta,z_n) + \rho^{0+}_{\Lambda_n}(\beta, z_n)\right] = \lim_{n\to\infty} \left[\rho^0_{\Lambda_n}(\beta,z_n) + \rho^{m+}_{\Lambda_n}(\beta, z_n)\right] $$
for any $m\in \mathbb{N}$. By replacing $L$ with $\Lambda_n$ in the second inequality in~(\ref{eq:convexineq2}) of Lemma~\ref{lem:convexineq}, and setting $z_1=z_n$ and $z_2=1$, we obtain
$$
\rho_{\Lambda_n}^{m+}(\beta,z_n) \leq \rho_{\Lambda_n}(\beta,1) \frac{z_n^{-1}- e^{-\beta E_m(\Lambda_n)}}{z_n^{-2}-e^{-\beta E_m(\Lambda_n)}}.
$$
It follows that for all $n$ and for all sufficiently large $m$,
$$
\frac{\rho_{\Lambda_n}^{m+}(\beta,z_n)}{\rho_{\Lambda_n}(\beta,1)} \leq  \frac{z_n^{-1}- e^{-\beta E_m(\Lambda_n)}}{z_n^{-2}-e^{-\beta E_m(\Lambda_n)}} \leq  \left(1-2\frac{E_0(\Lambda_n)}{E_m(\Lambda_n)}\right)^{-1} \leq  \left(1-2 C m^{1-\frac{d_s}{2}}\right)^{-1},
$$
where the Weyl asymptotics of the Laplacian was used in the last line. Thus
$$ \rho_{\rm tot} \leq \limsup_{n\to\infty}\left[\rho^0_{\Lambda_n}(\beta,z_n) + \rho_{\Lambda_n}(\beta,1) \left(1-2Cm^{1-\frac{d_s}{2}}\right)^{-1}  \right].$$
As the left-hand side is independent of $m$, we can take the limit $m\to\infty$ to complete the proof.
\qed \end{proof}

\begin{remark}
Here we encounter a subtlety not seen in the Euclidean setting, namely, that the length scale enters into the log-periodic modulation $G_0$ in the particle density. Therefore the thermodynamic limit needs to be handled with care. Theorem~\ref{thm:BEC} says that in the $\limsup$ density sense, condensation occurs if $\rho_{\rm tot}$ exceeds the upper critical density $\overline\rho_c(\beta)$. This appears to be the optimal threshold as our next result (Theorem~\ref{thm:BECfreeenergy}) suggests. Notice also that neither the condensate density $\rho_{\Lambda_n}^0(\beta,z_n)$ nor the critical density $\rho_{\Lambda_n}(\beta,1)$ likely has a limit as $n\to\infty$, but the sum of the two converges to $\rho_{\rm tot}$.

One might ask whether a stronger sense of convergence could be used to characterize the thermodynamic limit on fractals. A possibility would be to use the Ces\`{a}ro averaged density in place of the $\limsup$ density. However incorporating the former into the convexity proof of Theorem~\ref{thm:BEC} appears difficult.
\end{remark}

We have now shown that for an ideal massive Bose gas in an unbounded GSC with $d_s>2$, there is macroscopic occupation of particles in the ground state at sufficiently low temperatures. It is well known that Bose-Einstein condensation is a \textbf{quantum phase transition}. To illustrate this point, we will show that in the thermodynamic limit $L\to\infty$, the \emph{free energy density} of the Bose gas
$$ f_L(\beta,z) := \frac{F_L(\beta,z)}{V_s(L)}  = -\frac{1}{\beta V_s(L)}\log \Xi_{L,\beta,z}$$
becomes non-analytic in an appropriate sense. This is a familiar criterion of phase transition in the physics literature. A more technical criterion, namely, the non-uniqueness of Gibbs (or KMS) states at low temperatures, will be discussed elsewhere.

A straightforward computation shows that as $L\to\infty$,
$$ f_L(\beta,z) = -\frac{1}{(4\pi)^{\frac{d_s}{2}} \beta^{\frac{d_s}{2}+1}\hat{G}_{0,0}} \sum_{m=1}^\infty z^m G_0\left(-\log\left(\frac{m\beta}{L^2}\right)\right) m^{-\left(\frac{d_s}{2}+1\right)} + o(1).$$
Taking the sequence of free energy densities along $\{\Lambda_n\}_n$ yields the following result.
\begin{theorem}
For each $\rho_{\rm tot}>0$, let $z_n$ be the unique root of $\rho_{\Lambda_n}(\beta,z_n) = \rho_{\rm tot}$.
\begin{enumerate}[label={(\roman*)},nolistsep]
\item If $\rho_{\rm tot} \leq \overline\rho_c(\beta)$, and $\overline z$ is the root of $\displaystyle\limsup_{n\to\infty}\rho_{\Lambda_n}(\beta,\overline z)=\rho_{\rm tot}$, then 
\begin{equation}
  \liminf_{n\to\infty} f_{\Lambda_n}(\beta,z_n) = \liminf_{n\to\infty} f_{\Lambda_n}(\beta,\overline z) = -\frac{1}{(4\pi)^{\frac{d_s}{2}} \beta^{\frac{d_s}{2}+1}}\frac{\max G_0}{\hat{G}_{0,0}} {\rm Li}_{\frac{d_s}{2}+1}(\overline z),
\end{equation}
where $\displaystyle {\rm Li}_s(z)=\sum_{n=1}^\infty \frac{z^n}{n^s}$ is the polylogarithm.
 \label{freeenergy1}
\item If $\rho_{\rm tot} > \overline\rho_c(\beta)$, then $\displaystyle \lim_{n\to\infty} \left[f_{\Lambda_n}(\beta,z_n) - f_{\Lambda_n}(\beta,1)\right]=0$. \label{freeenergy2}
\end{enumerate}
\label{thm:BECfreeenergy}
\end{theorem}

\begin{proof}
The key is to realize that $-f_L(\beta,z) \geq 0$, and that $z\mapsto -f_L(\beta,z)$ is increasing and convex. From this it follows that for any $z_1>z_2$,
\begin{equation}
\frac{f_L(\beta,z_2) - f_L(\beta,z_1)}{z_1-z_2} \leq -\frac{\partial f_L}{\partial z}(\beta,z_1) = \frac{1}{\beta z_1}\rho_L(\beta,z_1).
\label{eq:freeenergyineq}
\end{equation}

For~\ref{freeenergy1} we replace $L$, $z_1$, and $z_2$ in~(\ref{eq:freeenergyineq}) with $\Lambda_n$, $z_n$, and $\overline z$ respectively to find
$$ 0\leq f_{\Lambda_n}(\beta,\overline z) - f_{\Lambda_n}(\beta, z_n) \leq \frac{\rho_{\rm tot}}{\beta}\left(1-\frac{\overline z}{z_n}\right).$$ 
Take the $\liminf$ on all sides and use Part~\ref{BECthm1} of Theorem~\ref{thm:BEC} to finish.

Similarly for~\ref{freeenergy2} we replace $L$, $z_1$, and $z_2$ in~(\ref{eq:freeenergyineq}) with $\Lambda_n$, $z_n$, and $1$ respectively. By virtue of Part~\ref{BECthm2} of Theorem~\ref{thm:BEC}, we can take the limit this time to deduce the result.
\qed \end{proof}

Part~\ref{freeenergy2} of the preceding theorem says that in the condensation regime the thermodynamic free energy density is, in a proper sense, independent of the particle density $\rho_{\rm tot}$, or of the (nonzero) condensate density. This is because in the thermodynamic limit, the condensate resides in the zero-energy eigenfunction and does not contribute to the free energy.  For example, in the sense of Ces\`{a}ro averaging, we have
$$ \lim_{n\to\infty} \frac{1}{n} \sum_{k=1}^n f_{\Lambda_k}(\beta,z_k) = \lim_{n\to\infty} \frac{1}{n} \sum_{k=1}^n f_{\Lambda_k}(\beta,1) = -\frac{1}{(4\pi)^{\frac{d_s}{2}} \beta^{\frac{d_s}{2}+1}}\zeta\left(\frac{d_s}{2}+1\right).$$
On the other hand, in the liminf sense
$$ \liminf_{n\to\infty}f_{\Lambda_n}(\beta,z_n) = \liminf_{n\to\infty} f_{\Lambda_n}(\beta,1) = -\frac{1}{(4\pi)^{\frac{d_s}{2}} \beta^{\frac{d_s}{2}+1}}\frac{\max G_0}{\hat{G}_{0,0}} {\rm Li}_{\frac{d_s}{2}+1}(1). $$
Combined with the result in Part~\ref{freeenergy1}, one may check that $\displaystyle \liminf_{n\to\infty} f_{\Lambda_n}(\beta,z_n)$, regarded as a function of $\rho_{\rm tot} \in (0,\infty)$, is non-analytic at $\overline\rho_c(\beta)$, since the left derivative (a nonzero constant times $\zeta\left(\frac{d_s}{2}\right)$) does not equal the right derivative ($0$).

\subsection{Connection to Brownian motion}

\begin{theorem}
Let $F$ be a GSC. Then the following are equivalent in the unbounded carpet $F_\infty$:
\begin{enumerate}
\item The spectral dimension $d_s>2$.
\item Brownian motion is transient.
\item Bose-Einstein condensation occurs for a low-temperature, high-density ideal Bose gas.
\end{enumerate}
\label{thm:equivBEC}
\end{theorem}

\begin{proof}
For $F_\infty$, the equivalence between (1) and (2) is given in~\cite{BB99}*{Theorem 8.1}. The equivalence between (1) and (3) is implied by Lemma~\ref{lem:rhoc} and Theorem~\ref{thm:BEC}.
\qed \end{proof}

On general graphs, the connection between spectral dimension and BEC has been mentioned in previous literature~\cite{BurioniCassiVezzani01,BurioniCassiVezzani02}. The equivalence between transient Brownian motion and BEC, on the other hand, appears to be a more general and robust criterion. To the author's knowledge, outside of $\mathbb{R}^d$ and $\mathbb{Z}^d$, an explicit statement about the latter has only been made by~\cite{BECMatsui, FidaleoGuidoIsola, Fidaleo} in the context of inhomogeneous graphs. In these works, the authors considered the adjacency operator, as opposed to the Laplacian, on various exotic graphs (\emph{e.g.} comb graphs, star graphs) and density-zero perturbations of Cayley trees. (We will discuss their model briefly in Section~\ref{subsec:BoseHubbard}.) They gave conditions for BEC by analyzing the spectral properties of the adjacency operator, in particular its Perron-Frobenius eigenvector, and showed that the existence of "hidden spectrum" implies the finiteness of the critical density. These developments lead us to believe that the deep connection between Brownian motion and BEC holds on very general spaces, including graphs, manifolds, and fractals. A careful study of KMS states on these spaces may uncover more details.

\begin{table}
\begin{tabular}{m{2.7cm}cccc}
&\includegraphics[width=0.14\textwidth]{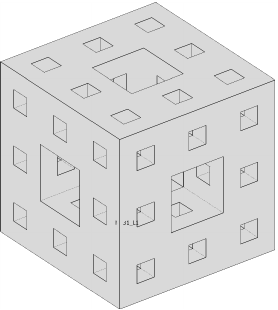}&
\includegraphics[width=0.14\textwidth]{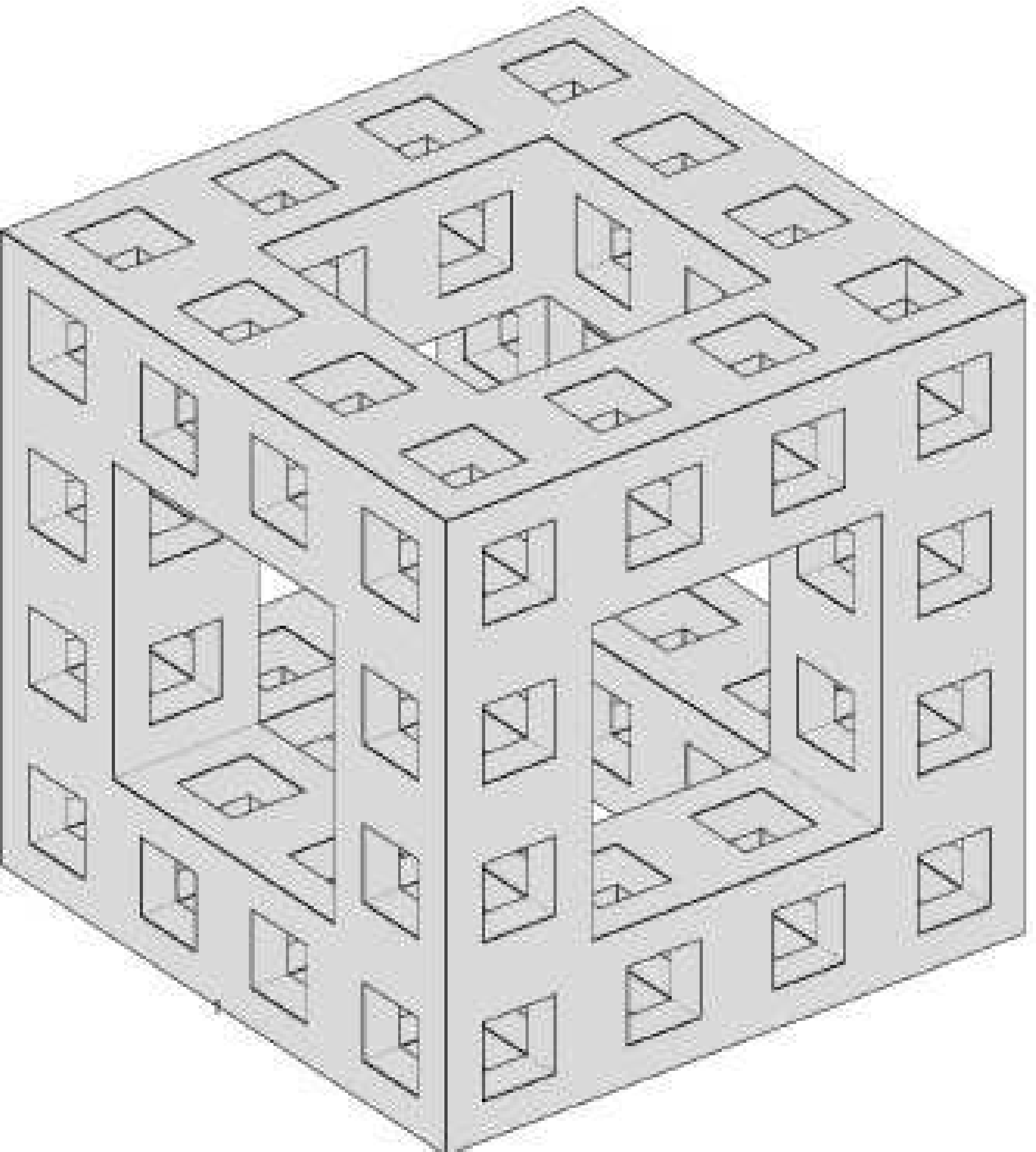}&
\includegraphics[width=0.14\textwidth]{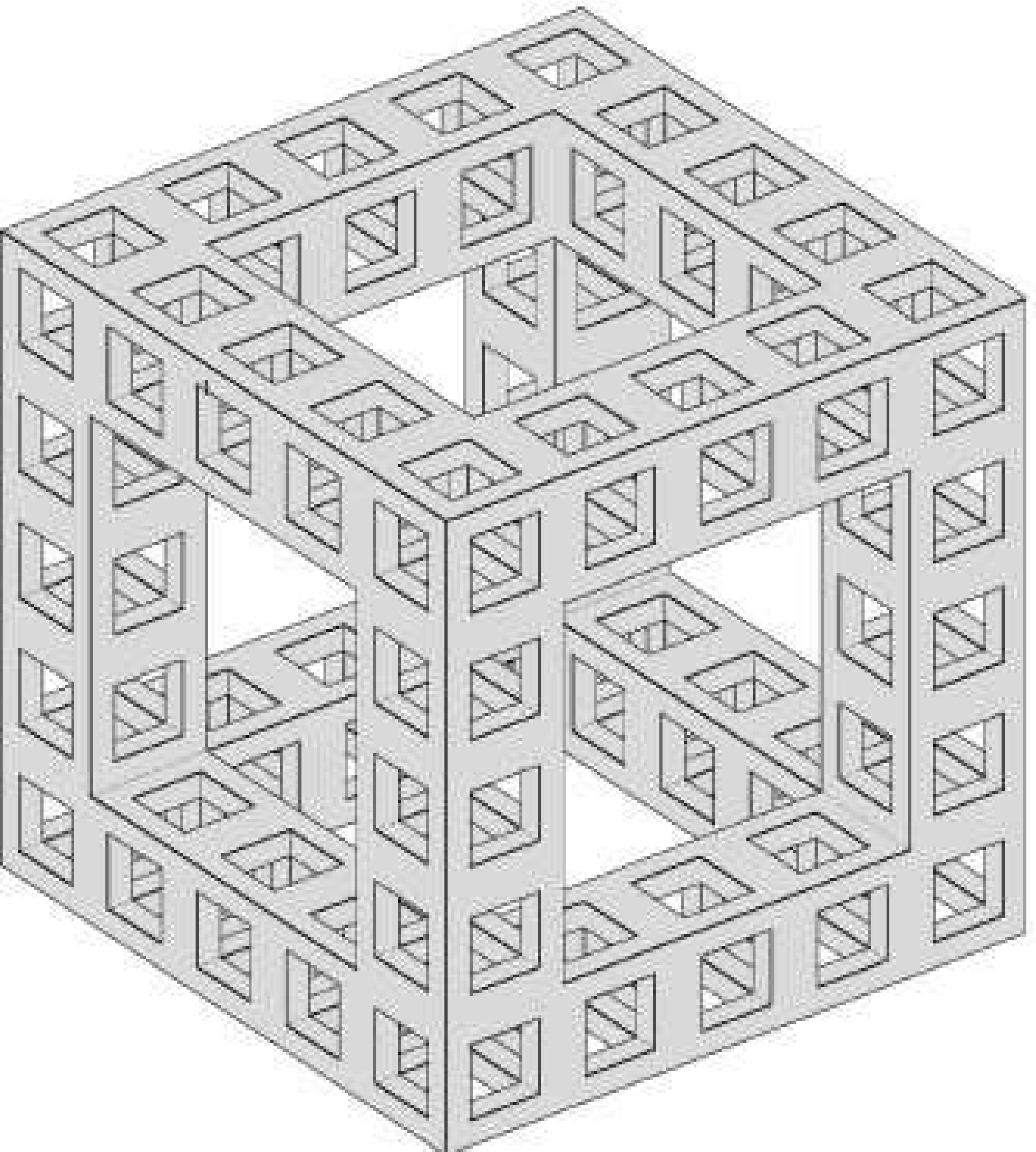}&
\includegraphics[width=0.14\textwidth]{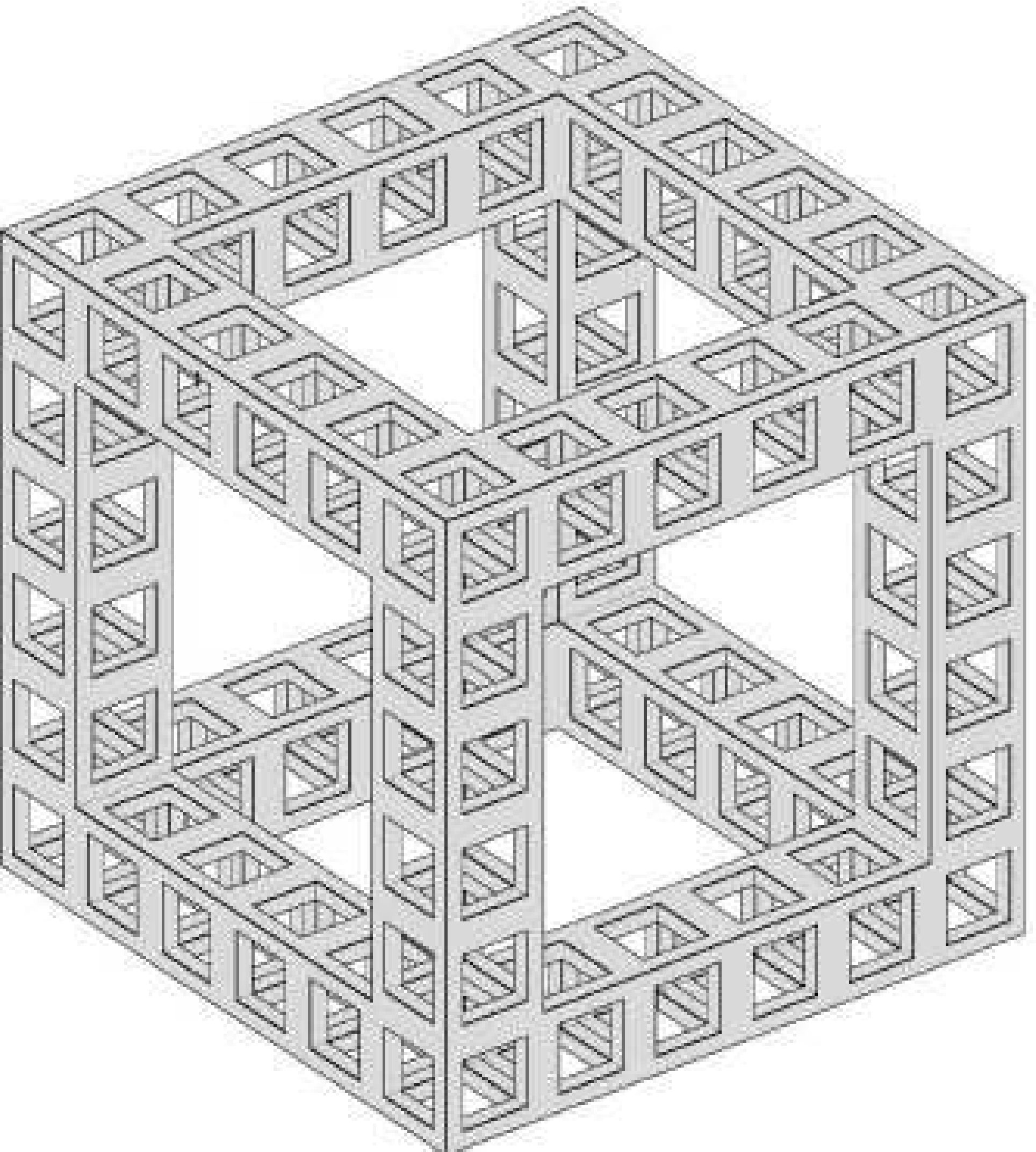}\\ 
& $MS(3,1)$ & $MS(4,2)$ & $MS(5,3)$ & $MS(6,4)$ \\ \hline
$d_h$ & $\log_3 20 \approx 2.73$ & $\log_4 32 = 2.5$ & $\log_5 44 \approx 2.35$ & $\log_6 56 \approx 2.25$ \\ \hline
Rigorous bounds on $d_s$~\begin{scriptsize}\cite{BB99}\end{scriptsize} & $2.21 \sim 2.60$ & $2.00\sim 2.26$ & ${\color{red} 1.89}\sim {\color{blue} 2.07}$  & $1.82 \sim 1.95$ \\ \hline
Numerical $d_s$~\begin{scriptsize}[CS]\end{scriptsize}& $2.51...$ & - & ${\color{blue} 2.01...}$ &  -\\ \hline
BEC exists? & Yes & Yes & {\color{blue} Yes (?)} & No \\
\vspace{5pt}
\end{tabular}
\caption{A sequence of Menger sponges whose spectral dimensions cross over $2$. The criterion for existence of BEC in an ideal Bose gas is determined by Theorem~\ref{thm:equivBEC}.}
\label{tab:sponge}
\end{table}

We end this section by giving specific examples where one may probe BEC in non-integer dimensions. Recall that if $F$ is a GSC embedded in $\mathbb{R}^d$, then $1 \leq d_s(F)<d_h(F) < d$. Also, an unbounded GSC whose cross-section contains a full copy of $\mathbb{R}_+^2$ supports transient Brownian motion, \emph{i.e.,} the carpet has $d_s>2$~\cite{BB99}. So if one's goal is to find a GSC whose spectral dimension is close to $2$, the most natural candidates are the Menger sponges in $\mathbb{R}^3$. For example, consider the family of Menger sponges $\{MS(a,a-2)\}_{a\geq 3}$, the first four members of which are shown in Table~\ref{tab:sponge}. All such sponges have $d_h>2$, but only $MS(3,1)$, $MS(4,2)$, and possibly $MS(5,3)$ have $d_s>2$, in which condensation of ideal Bose gas is possible. In particular $MS(5,3)$ lies just a hair above the critical dimension $2$, as our numerics show. It would be very interesting if one could find a one-parameter family of fractal sponges whose spectral dimensions vary across $2$, and study the onset of BEC in such spaces.

\section{Blackbody radiation in Sierpinski carpets} \label{sec:photon}

In the next two sections we discuss the Bose gas with $H=\sqrt{-\Delta}$ and $\mu=0$. This Hamiltonian corresponds to that of a gas of massless, non-interacting, relativistic particles, the prime example being photons. Here we have set $\mu=0$ because the particle number cannot be kept finite in realistic experiments. The grand canonical partition function $\Xi_\beta$ in this case reads
\begin{equation}\log \Xi_\beta= -{\rm Tr}_{\mathcal{H}_1} \log(1-e^{-\beta \underline{D}}) ,\end{equation}
where $\underline{D}=\sqrt{-\underline\Delta}$ is the dimensionful Dirac operator, carrying units of inverse length. 

\begin{proposition}
The grand canonical partition function $\Xi_\beta$ for an ideal photon gas in a bounded domain $F$ is given by
\begin{equation}
\log\Xi_\beta = \frac{1}{\pi i}\int_{\sigma-i\infty}^{\sigma+i\infty} \beta^{-2t} \Gamma(2t) \zeta(2t+1)\zeta_{\underline{\Delta}}(t,0) dt, \quad \sigma > \frac{d_s}{2}.
\label{eq:PhotonPF}
\end{equation}
\end{proposition}
The proof is essentially identical to that of Proposition~\ref{prop:Bosepart}, so we omit it.

\begin{lemma}
Let $F$ be a GSC. Then the grand canonical partition function $\Xi_\beta$ for an ideal photon gas confined in $LF$ is given by
\begin{equation}
\log\Xi_{L,\beta}=2\left(\sum_{k=0}^{d} \sum_{p\in\mathbb{Z}}\right)' \left(\frac{L}{\beta}\right)^{d_{k,p}}\frac{\Gamma\left(d_{k,p}\right)}{\Gamma\left(d_{k,p}/2\right)} \zeta\left(d_{k,p}+1\right) \hat{G}_{k,p} + \frac{\beta}{2} \zeta_{\underline\Delta_L}\left(-\frac{1}{2},0\right),
\label{eq:GSCPhotonPF}
\end{equation} 
where $\left(\sum_{k=0}^d \sum_{p\in\mathbb{Z}}\right)'$ means the double sum excluding the term $(k,p)=(d,0)$.
\label{lem:photonPF}
\end{lemma}
\begin{proof}
We evaluate~(\ref{eq:PhotonPF}) via a residue calculation. The first term of~(\ref{eq:GSCPhotonPF}) comes from the residues at the poles of $\zeta_{\underline\Delta}(s,0)$, though we exclude $d_{d,0}=0$ because it has zero residue (as well as being a simple pole of $\Gamma(2t)$ and of $\zeta(2t+1)$ simultaneously). Then there are contributions coming from the nonzero poles of $\Gamma(2t)$, which are $-\mathbb{N}/2$. But $\zeta_{\underline\Delta}(t,0)=0$ when $t \in -\mathbb{N}$ by Corollary~\ref{cor:polesresidues}, while $\zeta(2t+1)=0$ when $t\in -\left(\mathbb{N}+\frac{1}{2}\right)$. Only the residue at $t=-1/2$ contributes and produces the second term of~(\ref{eq:GSCPhotonPF}).
\qed \end{proof}

The first term of ({\ref{eq:GSCPhotonPF}), whose summands are all proportional to some nonnegative power of $L/\beta$, represent thermal (positive-temperature) contributions from the various codimension-$k$ spectral volumes. On the other hand, the second term, which is linearly proportional to $\beta/L$, has no thermal origins. To tackle this issue, let us add to~(\ref{eq:GSCPhotonPF})
\begin{equation}\log\Xi_{\rm Cas} := -\frac{\beta}{2} \zeta_{\underline{\Delta}}\left(-\frac{1}{2},0\right),\end{equation}
so we retain only the thermal contributions. Notice that this effectively adds a (free) energy
\begin{equation} E_{\rm Cas} = -\beta^{-1}\log \Xi_{\rm Cas} = \frac{1}{2}\zeta_{\underline\Delta} \left(-\frac{1}{2},0\right). \label{eq:CasimirEnergy}\end{equation}
to the existing vacuum energy $H_0$, which was $0$ by default.~(\ref{eq:CasimirEnergy}) is the renowned \emph{Casimir energy} representing the energy of vacuum fluctuations. It is half the regularized sum of the eigenvalues of the Dirac operator, where regularization should be understood as the analytic continuation of $\zeta_{\underline\Delta}(s,0)$ to $s=-1/2$. By Theorem~\ref{thm:meroext}, such continuation is always possible on GSCs. The Casimir energy has measurable consequences and will be the focus of the next section.

The most important phenomenon associated with thermal photon gas is \emph{blackbody radiation}. Here we are interested in the energy per spectral volume
$$\mathcal{E}(L,\beta) := \frac{\omega_{L,\beta,0}({\bf H})}{V_s(L)} = -\frac{1}{V_s(L)}\frac{\partial}{\partial \beta}\log\Xi_{L,\beta,0},$$
as well as the isotropic pressure
$$P(L,\beta) := \frac{\partial}{\partial V_s(L)}F_{L,\beta,0} = -\frac{1}{\beta}\frac{\partial}{\partial V_s(L)} \log\Xi_{L,\beta,0}.$$
\begin{proposition}
Let $F$ be a GSC. As $L\to\infty$, the mean energy density and isotropic pressure of a thermal photon gas inside $LF$ are
\begin{eqnarray}
\mathcal{E}(L,\beta) &=& \beta^{-(d_s+1)} H_1\left(-\log\left(\frac{\beta}{2L}\right)\right) +o (1), \label{eq:photonenergy}\\
P(L,\beta) &=& \frac{\mathcal{E}(L,\beta)}{d_s}, \label{eq:photonpressure}
\end{eqnarray}
where $H_1(\cdot)$ is $\frac{1}{2}\log R$-periodic with
\begin{eqnarray}
\nonumber H_1(x) &=&  \frac{d_s}{\pi^{\frac{d_s+1}{2}}}\Gamma\left(\frac{d_s+1}{2}\right) \zeta(d_s+1) \label{eq:H1}  \\
&& + \frac{1}{\pi^{\frac{d_s+1}{2}} \hat{G}_{0,0}} \sum_{p \in \mathbb{Z} \backslash \{0\}} \left[e^{\frac{4\pi p i x}{\log R}} \Gamma\left(\frac{d_{0,p}+1}{2}\right) \zeta(d_{0,p}+1) d_{0,p} \hat{G}_{0,p}\right].
\end{eqnarray}
\label{prop:BBEnergy}
\end{proposition}

\begin{proof}
Straightforward differentiations of the thermal terms in~(\ref{eq:GSCPhotonPF}), combined with the gamma function identity
$$ \frac{\Gamma(z)}{\Gamma(z/2)} = \frac{\Gamma((z+1)/2)}{2^{1-z}\sqrt\pi},$$ yield the results. We note that for all $s \in \mathbb{R}$, $|\Gamma(s+it)| \sim \sqrt{2\pi}|t|^{s-1/2}e^{-\pi |t|/2}$ as $t\to\pm\infty$ by Stirling's formula, while for $s>1$ and $t\in \mathbb{R}$ we have $|\zeta(s+it)| \leq \zeta(s)$. These observations, plus the boundedness of the $|\hat{G}_{0,p}|$, are sufficient to ensure the summability of~(\ref{eq:H1}).
\qed \end{proof}

We have intentionally split $H_1$ in order to make comparisons with the Euclidean case. If only the first term in~(\ref{eq:H1}) was considered, then~(\ref{eq:photonenergy}) would be a direct extension of the Stefan-Boltzmann law from Euclidean to noninteger dimensions. But the discrete scale invariance of the GSC gives rise to log-periodic modulations depending on the temperature as well, \emph{viz.} the second term in~(\ref{eq:H1}). Note also that the $\frac{1}{2}$ in the periodicity $\frac{1}{2}\log R$ of $H_1$ comes from the Dirac operator being the square root of the Laplacian.

Meanwhile,~(\ref{eq:photonpressure}) represents the fractal version of the \emph{equation of state} for the photon gas~\cite{FractalThermo}. This is to be compared with the Euclidean version $P =\mathcal{E}/d$, where $d$ is the Euclidean dimension (also the spectral dimension, since $d_w=2$).

\section{Casimir effect in Sierpinski carpet waveguide}\label{sec:Casimir}

In this section we derive the radiation pressure on some boundary faces of a compact fractal domain, due to either vacuum (zero-temperature, $\beta=\infty$) or thermal (nonzero-temperature, $\beta < \infty$) effects. 

\subsection{The parallel-plate setup}

\begin{figure}
\centering
\includegraphics[width=0.4\textwidth]{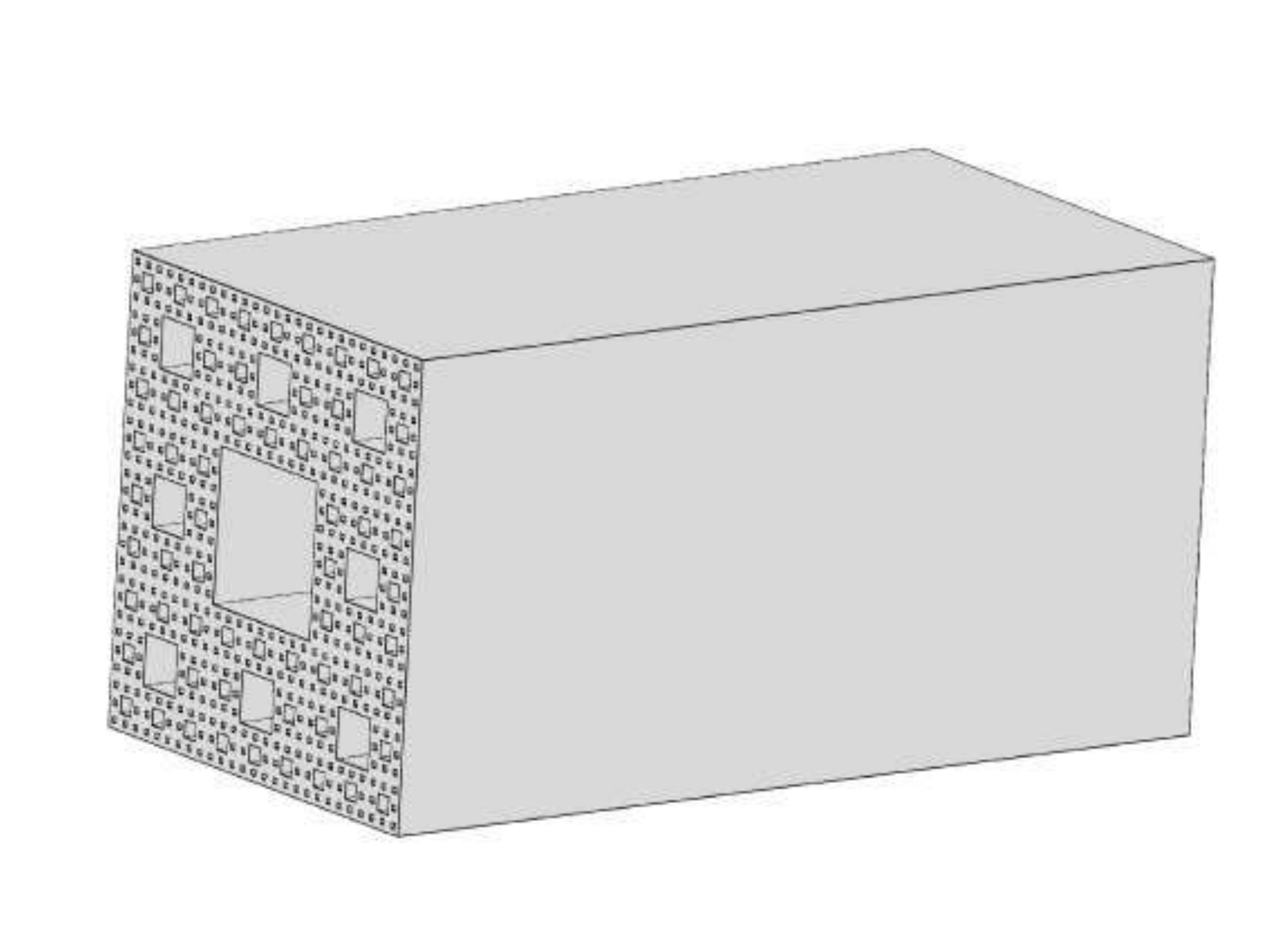}
\caption{The Sierpinski carpet waveguide $a(GSC)\times [0,b]$.}
\label{fig:SCWaveguide}
\end{figure}

The original calculation, due to H. Casimir in 1948, concerned the vacuum energy in $[0,a]^2 \times [0,b]$ for fixed $a \gg 1$, and the resultant pressure on the two parallel faces $[0,a]^2 \times \{0\}$ and $[0,a]^2 \times \{b\}$ as a function of $b$. The conventional method is to carry out a sum over the (exactly known) eigenvalues of the Dirac operator, take the $a\to\infty$ limit, then apply "zeta-regularization" to the formally divergent sum to produce a finite value. Similar calculations have also been done on other Euclidean domains or manifolds, such as spheres and cylinders.

In the spirit of the classical parallel-plate setup, we will work with a Sierpinski carpet waveguide $\Omega_{a,b}:=a(GSC) \times [0,b]$, where we take $GSC \subset \mathbb{R}^2$ to be some 2-dimensional Sierpinski carpet and fix $a\gg 1$. See Figure~\ref{fig:SCWaveguide}. The parallel plates in this case correspond to $a(GSC) \times \{0\}$ and $a(GSC) \times \{b\}$. For consistency, we will impose Dirichlet conditions on the outer boundaries $\partial ([0,a]^2) \times [0,b]$, $a(GSC) \times \{0\}$ and $a(GSC)\times\{b\}$, and Neumann conditions on the inner surfaces. The Hausdorff and spectral dimensions of the waveguide are, respectively, $d_h(\Omega) = d_h(GSC) +1 \in (2,3)$ and $d_s(\Omega) = d_s(GSC)+1 \in (2, d_h(\Omega))$.

\subsection{Casimir pressure at zero temperature}

Our main objective is to compute the spectral zeta function $\zeta_\Omega\left(-\frac{1}{2},0\right)$, or twice the Casimir energy. Because the Laplacian on $\Omega$ is the direct sum of the Laplacian on $GSC$ and the Laplacian on $[0,b]$, the associated heat kernel trace factorizes:
$$K_{\Omega_{a,b}}(t) = K_{a(GSC)}(t) K_{[0,b]}(t) = K_{GSC}\left(a^{-2}t\right) K_{[0,1]}\left(b^{-2}t\right).$$
Recalling that the spectrum of the Dirichlet Laplacian on the unit interval is $\bigcup_{j=1}^\infty \{j^2 \pi^2\}$, we find
\begin{eqnarray*}
K_{[0,1]}(t) &=& \sum_{j=1}^\infty e^{-j^2 \pi^2 t} 
=\frac{1}{2}\left[\sum_{j\in \mathbb{Z}} e^{-j^2 \pi^2 t}-1\right]\\
&=& \frac{1}{2}\left[\frac{1}{\sqrt{\pi t}} \sum_{j\in\mathbb{Z}} e^{-j^2/t}-1\right]=\frac{1}{\sqrt{4\pi t}} - \frac{1}{2} +\frac{1}{\sqrt{\pi t}}\sum_{j=1}^\infty e^{-j^2/t},
\end{eqnarray*}
where the Poisson summation formula was employed in the third equality. Then using~(\ref{eq:HK}) with $d=2$ for $K_{GSC}(t)$, we get
\begin{eqnarray}
\nonumber K_{\Omega_{a,b}} (t) &=&\left[\sum_{k=0}^2 \sum_{p\in\mathbb{Z}}\hat{G}_{k,p} \left(\frac{t}{a^2}\right)^{-d_{k,p}/2} + \mathcal{O} \left(\exp\left(-C\left(\frac{t}{a^2}\right)^{-\frac{1}{d_w-1}}\right)\right)\right] 
\\
&&\times\left[\frac{b}{\sqrt{4\pi t}} - \frac{1}{2} +\frac{b}{\sqrt{\pi t}}\sum_{j=1}^\infty e^{-j^2 b^2/t}\right].
\label{eq:WaveguideHKT}
\end{eqnarray}
The spectral zeta function for the Laplacian on $\Omega_{a,b}$ is the meromorphic extension of
\begin{equation}
\zeta_{\Omega_{a,b}}(s,\gamma) = \frac{1}{\Gamma(s)} \int_0^\infty t^s K_{\Omega_{a,b}}(t) e^{-\gamma t} \frac{dt}{t},\quad {\rm Re}(s)> \frac{d_s(\Omega)}{2}.
\label{eq:waveguidezeta}
\end{equation}

\begin{theorem}
As $a\to\infty$,
\begin{eqnarray}
\nonumber &&\Gamma\left(-\frac{1}{2}\right)\zeta_{\Omega_{a,b}}\left(-\frac{1}{2},0\right) \\ &=&\frac{1}{\sqrt{\pi}} \sum_{k=0}^2  \sum_{p\in\mathbb{Z}}\hat{G}_{k,p} \left(\frac{a}{b}\right)^{d_{k,p}} b^{-1} \zeta\left(2+d_{k,p}\right) \Gamma\left(1+\frac{d_{k,p}}{2}\right)+ o(1) .
\end{eqnarray}
\end{theorem}

\begin{proof}
Though the proof is similar to that of Theorem~\ref{thm:meroext}, we will be more precise in the computation here. Let $f(t) = \mathcal{O}\left(\exp\left(-Ct^{-\frac{1}{d_w-1}}\right)\right)$. By definition, there exist $T,M \in (0,\infty)$ such that $|f(t)| \leq M \exp\left(-Ct^{-\frac{1}{d_w-1}}\right)$ for all $t \in (0,T)$. It follows that $|f(a^{-2} t)| \leq M\exp\left(-Ca^{\frac{2}{d_w-1}}t^{-\frac{1}{d_w-1}}\right)$ for all $t \in (0,a^2 T)$. So when computing the Mellin integral we split the domain of integration $(0,\infty)$ into two intervals, $(0,a^2 T)$ and $[a^2 T,\infty)$.

Expanding the product in~(\ref{eq:WaveguideHKT}) and then plugging it into~(\ref{eq:waveguidezeta}), we find 
$$\zeta_{\Omega_{a,b}}\left(-\frac{1}{2},\gamma\right) \Gamma\left(-\frac{1}{2}\right)  = \int_0^{a^2 T} I(t)e^{-\gamma t} \frac{dt}{t} + \int_{a^2 T}^\infty I(t) e^{-\gamma t}\frac{dt}{t},$$
where $I(t)$ contains terms of the following types:
\begin{enumerate}
\item Polynomial terms of the form $t^{-p}$, ${\rm Re}(p)>0$.
\item Product of a polynomial term and an exponential term: $t^{-p}\exp(-c_5 t^{-\eta})$, ${\rm Re}(p), \eta>0$.
\item Product of a polynomial term and two exponential terms: \\$t^{-p} \exp(-c_6 t^{-\eta})\exp(-c_7 t^{-1})$, ${\rm Re}(p),\eta>0$.
\end{enumerate}

For the polynomial terms we have
$$\int_0^{a^2 T} t^{-p} e^{-\gamma t} \frac{dt}{t} = \gamma^p \boldsymbol{\gamma}(-p, \gamma a^2 T),\quad \int_{a^2 T}^\infty t^{-p} e^{-\gamma t} \frac{dt}{t} =\gamma^p \Gamma(-p, \gamma a^2 T), $$
where $\boldsymbol{\gamma}(s,x)$ and $\Gamma(s,x)$ are, respectively, the lower and upper incomplete gamma functions. Then we need to use the asymptotics of the incomplete gamma functions near the branch point $x=0$: $x^{-s} \boldsymbol{\gamma}(s,x) \to 1/s$ and, for ${\rm Re}(s)<0$, $x^{-s} \Gamma(s,x) \to -1/s$. These imply that
$$ \lim_{\gamma \to 0}\left(\int_0^{a^2 T} t^{-p} e^{-\gamma t} \frac{dt}{t}+ \int_{a^2 T}^\infty t^{-p} e^{-\gamma t} \frac{dt}{t}\right) =0.$$

Next we tackle terms of type (2). By making a change of variables we find
$$
J_1(a,b,\gamma):=\int_0^{a^2 T} t^{-p} e^{-c_5 t^{-\eta}} e^{-\gamma t}\frac{dt}{t}= \frac{1}{\eta c_5^{p/\eta}} \int_{c_5/(a^2 T)^\eta}^\infty u^{p/\eta} e^{-u} e^{-\gamma(c_5/u)^{1/\eta}} \frac{du}{u}.
$$
Notice that for $u \in (c_5/(a^2 T)^\eta,\infty)$ we have the upper bound $\left|e^{-\gamma(c_5/u)^{1/\eta}}\right| \leq e^{|\gamma|a^2 T}$, and the remaining integrand is integrable:
$$ \left|\int^\infty_{c_5/(a^2 T)^\eta} u^{p/\eta}e^{-u}\frac{du}{u} \right|= \left|\Gamma\left(\frac{p}{\eta},\frac{c_5}{(a^2 T)^\eta}\right)\right| <\infty.$$
So by the dominated convergence theorem,
$$\lim_{\gamma\to 0}J_1(a,b,\gamma) =\frac{1}{\eta c_5^{p/\eta}} \int_{c_5/(a^2 T)^\eta}^\infty u^{p/\eta} e^{-u}\frac{du}{u} = \frac{1}{\eta c_5^{p/\eta}}  \Gamma\left(\frac{p}{\eta},\frac{c_5}{(a^2 T)^\eta}\right). $$
The same argument applies to
$$ J_2(a,b,\gamma) := \int_{a^2 T}^\infty t^{-p} e^{-c_5t^{-\eta}} e^{-\gamma t} \frac{dt}{t},$$
yielding
$$\lim_{\gamma\to 0} J_2(a,b,\gamma) =\frac{1}{\eta c_5^{p/\eta}}\boldsymbol{\gamma}\left(\frac{p}{\eta},\frac{c_5}{(a^2 T)^\eta}\right).$$
Thus
\begin{eqnarray*}
&& \lim_{\gamma\to0}\left(\int_0^{a^2 T} t^{-p} e^{-c_5 t^{-\eta}} e^{-\gamma t}\frac{dt}{t}+\int_{a^2 T}^\infty t^{-p} e^{-c_5 t^{-\eta}} e^{-\gamma t}\frac{dt}{t} \right) \\ &=& \frac{1}{\eta c_5^{p/\eta}} \left[ \Gamma\left(\frac{p}{\eta}, \frac{c_5}{(a^2 T)^\eta}\right)+ \boldsymbol{\gamma}\left(\frac{p}{\eta}, \frac{c_5}{(a^2 T)^\eta}\right)\right] \\ &=& \frac{1}{\eta c_5^{p/\eta}} \Gamma\left(\frac{p}{\eta}\right).
\end{eqnarray*}

Finally, for terms of type (3) we have
\begin{eqnarray*}
&& \left|\int_0^{a^2 T} t^{-p} e^{-c_6 t^{-\eta}} e^{-c_7 t^{-1}} e^{-\gamma t} \frac{dt}{t} \right| \\&=&  
\left|\frac{1}{c_7^p} \int_{c_7/(a^2 T)}^\infty u^p e^{-u} e^{-c_6(u/c_7)^\eta} e^{-\gamma(c_7/u)} \frac{du}{u}\right| \\
& \leq & \frac{e^{|\gamma|c_7 a^2 T} e^{-c_7/(a^2 T)}}{c_7^{{\rm Re}(p)}} \int_{c_7/(a^2 T)}^\infty u^{{\rm Re}(p)} e^{-c_6(u/c_7)^\eta} \frac{du}{u}\\
&=&  \frac{e^{|\gamma|c_7 a^2 T} e^{-c_7/(a^2 T)}}{\eta c_6^{{\rm Re}(p)/\eta}} \Gamma\left(\frac{{\rm Re}(p)}{\eta},\frac{c_6}{(a^2 T)^\eta}\right).
\end{eqnarray*}
Thus
$$\lim_{\gamma\to 0} \left|\int_0^{a^2 T} t^{-p} e^{-c_6 t^{-\eta}} e^{-c_7 t^{-1}} e^{-\gamma t} \frac{dt}{t} \right| \leq \frac{e^{-c_7/(a^2 T)}}{\eta c_6^{{\rm Re}(p)/\eta}} \Gamma\left(\frac{{\rm Re}(p)}{\eta},\frac{c_6}{(a^2 T)^\eta}\right).$$
A similar calculation yields
$$\lim_{\gamma\to 0} \left|\int_{a^2 T}^\infty t^{-p} e^{-c_6 t^{-\eta}} e^{-c_7 t^{-1}} e^{-\gamma t} \frac{dt}{t} \right| \leq \frac{1}{\eta c_6^{{\rm Re}(p)/\eta}} \boldsymbol{\gamma}\left(\frac{{\rm Re}(p)}{\eta},\frac{c_6}{(a^2 T)^\eta}\right).$$

After putting in the numbers we obtain
\begin{eqnarray*}
\zeta_{\Omega_{a,b}}\left(-\frac{1}{2},0\right)\Gamma\left(-\frac{1}{2}\right) &=& \frac{1}{\sqrt{\pi}}\sum_{k=0}^2 \sum_{p\in\mathbb{Z}} \hat{G}_{k,p} \frac{a^{d_{k,p}}}{b^{d_{k,p}+1}}  \sum_{j=1}^\infty \frac{1}{j^{(d_{k,p}+2)}}\Gamma\left(\frac{d_{k,p}}{2}+1\right) \\ &&+ I_4(a,b),
\end{eqnarray*}
where
\begin{eqnarray*}
|I_4(a,b)| \leq \frac{M}{2}\frac{d_w -1}{a C^{(d_w-1)/2}} \Gamma\left(\frac{d_w-1}{2},\frac{C}{T^{1/(d_w-1)}}\right) + I_5(a,b)
\end{eqnarray*}
and $\displaystyle \lim_{a\to\infty}I_5(a,b)=0$. This concludes the proof.
\qed \end{proof}

The Casimir pressure is defined to be the force per unit spectral area on each parallel plate as the separation $b$ of the two plates is varied, that is,
$$P_{\rm Cas}(a,b) := -\frac{1}{V_s(a(GSC))} \frac{\partial }{\partial b}E_{\rm Cas}(a,b).$$
The following results are corollaries of the theorem.
\begin{proposition}
Denote 
\begin{equation} \mathcal{C}_p := \zeta\left(1+d_s(\Omega)+\frac{4\pi pi}{\log R}\right) \Gamma\left(\frac{1}{2}+\frac{d_s(\Omega)}{2} + \frac{2\pi p i}{\log R}\right)\hat{G}_{0,p}.
\end{equation}
As $a\to\infty$, the Casimir energy in a SC waveguide $\Omega_{a,b}$ is
\begin{eqnarray}
\nonumber E_{\rm Cas}(a,b) &:=&\frac{1}{2}\zeta_{\Omega_{a,b}}\left(-\frac{1}{2},0\right)\\
&=& -\frac{1}{4\pi} \frac{a^{d_s(GSC)}}{b^{d_s(GSC)+1}} \sum_{p\in\mathbb{Z}} \mathcal{C}_p e^{\frac{4\pi pi}{\log R}\log\left(\frac{a}{b}\right)}+ \mathcal{O}(a^{\frac{2}{d_w(GSC)}}).
 \end{eqnarray}
The zero-temperature Casimir pressure on each parallel plate $a(GSC)\times\{0\}$ or $a(GSC)\times\{b\}$ is
\begin{eqnarray}
\nonumber P_{\rm Cas}(a,b) &=&  -\frac{b^{-(d_s(\Omega)+1)}}{(4\pi)^{\frac{d_s(\Omega)+1}{2}} \hat{G}_{0,0}} \sum_{p\in\mathbb{Z}} \mathcal{C}_p \left[d_s(\Omega) + \frac{4 \pi  pi}{\log R}\right] e^{\frac{4\pi p i}{\log R} \log\left(\frac{a}{b}\right)}\\
 &&+\mathcal{O}\left(a^{-2\frac{d_h(GSC)-1}{d_w(GSC)}}\right).
\end{eqnarray}
\label{prop:CasimirPressure}
\end{proposition}

The key take-away message is that the Casimir pressure scales with the separation of the plates $b$ as
\begin{equation}
P_{\rm Cas}(a,b) = b^{-(d_s(\Omega)+1)}H_2\left(\log\left(\frac{a}{b}\right)\right) + o(1) \quad \mbox{as}~a\to\infty,
\end{equation}
where $H_2(\cdot)$ is $\frac{1}{2}\log R$-periodic. We do not know the sign of $P_{\rm Cas}$ as it requires more information about the function $G_0$ than currently available. 

\subsection{Casimir pressure at positive temperature}

To calculate the Casimir pressure due to thermal effects, we first derive the partition function, then find the free energy and take appropriate derivatives. Since most of the arguments resemble those in the blackbody radiation calculation (Section \ref{sec:photon}), we will omit the computational details.

\begin{lemma}
The simple poles of the spectral zeta function $\zeta_{\Omega_{a,b}}(\cdot,0)$ are contained in the set
\begin{equation}
\bigcup_{k=0}^2 \bigcup_{p\in\mathbb{Z}} \left( \left\{\frac{d_{k,p}+1}{2}\right\} \bigcup \left\{\frac{d_{k,p}}{2}\right\}\right),
\end{equation}
where the residues at the two sets of poles read, respectively,
\begin{equation}
\frac{a^{d_{k,p}}b}{\sqrt{4\pi}}\frac{\hat{G}_{k,p}}{\Gamma\left(\frac{d_{k,p}+1}{2}\right)},\quad -\frac{1}{2} a^{d_{k,p}} \frac{\hat{G}_{k,p}}{\Gamma(d_{k,p}/2)}.
\end{equation}
\end{lemma}
\begin{proof}
It is straightforward to check that only the polynomial terms in the expansion of $K_{\Omega_{a,b}}(t)$~(\ref{eq:WaveguideHKT}) produce singularities. The rest of the proof is as outlined in Theorem~\ref{thm:meroext} and Corollary~\ref{cor:polesresidues}.
\qed \end{proof}

\begin{lemma}
The thermal terms in the grand canonical partition function \\ $\Xi_{\Omega_{a,b},\beta}$ for the photon gas in the Sierpinski carpet waveguide $\Omega_{a,b}$ are given by
\begin{eqnarray}
\nonumber \log \Xi_{\Omega_{a,b},\beta} &=& \frac{1}{\sqrt\pi}\sum_{k=0}^2 \sum_{p\in\mathbb{Z}} \frac{a^{d_{k,p}}b}{\beta^{d_{k,p}+1}}\frac{\Gamma(d_{k,p}+1)}{\Gamma\left(\frac{d_{k,p}+1}{2}\right)}\zeta(d_{k,p}+2) \hat{G}_{k,p} \\
&& - \left(\sum_{k=0}^2 \sum_{p\in\mathbb{Z}}\right)' \frac{a^{d_{k,p}}}{\beta^{d_{k,p}}} \frac{\Gamma(d_{k,p})}{\Gamma(d_{k,p}/2)} \zeta(d_{k,p}+1)\hat{G}_{k,p},
\end{eqnarray}
where $\left(\sum_{k=0}^d \sum_{p\in\mathbb{Z}}\right)'$ means the double sum excluding the term $(k,p)=(d,0)$.
\end{lemma}
\begin{proof}
The arguments are identical to those in the proof of Lemma~\ref{lem:photonPF}.
\qed \end{proof}

The Casimir pressure is then defined to be
$$ P_{{\rm Cas},\beta}(a,b):= -\frac{1}{V_s(a(GSC))} \frac{\partial}{\partial b} F_{\Omega_{a,b},\beta}= \frac{1}{\beta V_s(a(GSC))} \frac{\partial}{\partial b} \log \Xi_{\Omega_{a,b},\beta}.$$

\begin{proposition}
As $a\to\infty$, the Casimir pressure on each parallel plate \\ $a(GSC)\times \{0\}$ or $a(GSC)\times\{b\}$ at inverse temperature $\beta$ reads
\begin{equation}
P_{{\rm Cas},\beta}(a,b) = \beta^{-(d_s(\Omega)+1)}H_3\left(-\log\left(\frac{\beta}{2a}\right)\right) + o(1),
\end{equation}
where $H_3(\cdot)$ is $\frac{1}{2}\log R$-periodic with
\begin{equation}
H_3(x) = \frac{\hat{G}_{0,p}/\hat{G}_{0,0}}{\pi^{\frac{d_s(\Omega)+1}{2}}}\sum_{p\in\mathbb{Z}} e^{\frac{4\pi i px}{\log R}} \Gamma\left(\frac{1}{2}+\frac{d_s(\Omega)}{2}+\frac{2\pi p i}{\log R}\right)\zeta\left(1+d_s(\Omega)+\frac{4\pi p i}{\log R}\right).
\end{equation}
\end{proposition}

Notice that the dominant term in the positive-temperature Casimir pressure is independent of the separation of plates $b$: like in blackbody radiation, the scaling is with respect to temperature.

This concludes our discussions of various results regarding ideal massive or massless Bose gas in Sierpinski carpets. We stress once again the power of the heat kernel and the zeta function technique in quantum statistical mechanics. Not only does it apply to very general spaces, it also give unambiguous answers to major thermodynamic questions, such as the emergence of the Casimir energy.

\section{Towards interacting Bose gas on Sierpinski carpet graphs}\label{sec:graph}

The above discussion of the ideal massive Bose gas is just the prelude to a full understanding of quantum many-body physics on fractals. This is because some of the most salient features of BEC (in $\mathbb{R}^d$ or $\mathbb{Z}^d$), such as superfluidity and soliton excitations, cannot be explained by the ideal Bose gas. So it behooves us to take the next step and consider models of interacting Bose gas on Sierpinski carpets. Mathematically speaking, proving BEC in an interacting Bose system is a challenging problem (see~\cite{BoseMathBook} for a careful exposition). In this section we introduce a particular interacting system which has been studied extensively on $\mathbb{Z}^d$, then focus on the so-called hardcore Bose gas, which we hope to study in more depth on Sierpinski carpet graphs.

\subsection{The Bose-Hubbard model} \label{subsec:BoseHubbard}

In what follows we have in mind a connected graph $G=(V,E)$ with bounded degree. To describe an ideal Bose gas on $G$, we use the second-quantized Hamiltonian in the grand canonical ensemble $d\Gamma(-\Delta -\mu {\bf 1})$, where $\Delta$ is the combinatorial graph Laplacian on $G$. To express this Hamiltonian, let $a_x^\dagger$ and $a_x$ be the boson creation and annihilation operators at $x \in V$, which satisfy the canonical commutation relations (CCR): $a_x a_y^\dagger - a_y^\dagger a_x=\delta_{xy}$ and $a_x a_y - a_y a_x=0$ for all $x,y \in V$. We also let $n_x= a_x^\dagger a_x$, which has the interpretation as the particle number at $x$. Then one may verify that
$$d\Gamma(-\Delta-\mu {\bf 1}) = -\frac{1}{2}\sum_{\langle xy \rangle\in E} (a_x^\dagger a_y + a_y^\dagger a_x) -\sum_{x\in V}(\mu - \deg(x)) n_x. $$

A slight variant of this Hamlitonian uses the adjacency operator in place of the graph Laplacian, \emph{i.e.,}
\begin{equation}
{\bf H}=-\frac{1}{2}\sum_{\langle xy \rangle\in E} (a_x^\dagger a_y + a_y^\dagger a_x) -\mu \sum_{x\in V} n_x .
\label{eq:adjacencyH}
\end{equation}
In a realistic experimental setup, multiple atoms tend not to occupy the same site $x \in V$ due to forbidden overlap of atomic orbitals. This can be modelled by introducing an on-site repulsion term to~(\ref{eq:adjacencyH}). The resultant Hamiltonian is that of the Bose-Hubbard model,
\begin{equation}
{\bf H}_{BH}=-\frac{J}{2} \sum_{\langle xy \rangle \in E} (a_x^\dagger  a_y+a_y^\dagger a_x) -\mu \sum_{x\in V} n_x + U \sum_{x \in V} n_x (n_x -1),
\label{eq:BoseHubbard}
\end{equation}
with parameters $\mu \in \mathbb{R}$ and $J,U>0$. The Hamiltonian~(\ref{eq:BoseHubbard}) is widely used by condensed matter theorists to model cold atomic systems in optical lattices. In the works of BEC on inhomogeneous graphs~\cite{BECMatsui, FidaleoGuidoIsola, Fidaleo}, the authors investigated a special case of~(\ref{eq:BoseHubbard}) with $\mu=0$ and $U=0$.

\subsection{Hardcore Bose gas and the XY model}

Let us focus on another special case of~(\ref{eq:BoseHubbard}) where the on-site repulsion $U=\infty$. In this scenario $n_x$ for any $x \in V$ can only take the value $0$ or $1$, that is, particles interact as if they were hard cores. The Hilbert space of this hardcore Bose gas is then $(\mathbb{C}^2)^{\otimes G}$, which is spanned by the tensor product of the basis vectors $\left\{\left(\begin{array}{c}1 \\ 0\end{array}\right)_x, \left(\begin{array}{c}0\\ 1\end{array}\right)_x\right\}$, representing respectively the occupation or non-occupation of vertex $x$. Naturally, the creation and annihilation operators at $x$ read
$$ a_x^\dagger = \left(\begin{array}{cc} 0 & 1 \\ 0 & 0 \end{array}\right), \quad a_z = \left(\begin{array}{cc} 0 & 0 \\ 1 & 0 \end{array}\right).$$
Notice that they do not satisfy the CCR anymore: $a_x a_y^\dagger - a_y^\dagger a_x=(1-2n_x) \delta_{xy}$.

We can express the Hamiltonian of the hardcore Bose gas as
\begin{equation}
{\bf H}_{hc} = -\frac{J}{2}\sum_{\langle xy \rangle\in E} (a_x^\dagger a_y +a_y^\dagger a_x)- \mu\sum_{x\in V}n_x.
\label{eq:hchamiltonian}
\end{equation}
Introduce the Pauli matrices
$$S^1 = \frac{1}{2}\left(\begin{array}{cc} 0 & 1 \\ 1 & 0 \end{array}\right),\quad S^2=\frac{1}{2}\left(\begin{array}{cc} 0 & -i \\ i & 0 \end{array}\right),\quad S^3 =\frac{1}{2}\left(\begin{array}{cc} 1 & 0 \\ 0 & -1\end{array}\right).$$
Then up to an additive constant~(\ref{eq:hchamiltonian}) becomes
\begin{equation}
{\bf H}_{hc} =  -J\sum_{\langle xy \rangle \in E} \left(S_x^1 S_y^1 + S_x^2 S_y^2\right) - \mu \sum_{x\in V}S_x^3,
\label{eq:hcXY}
\end{equation}
which is the Hamiltonian of the \textbf{quantum (ferromagnetic) XY model}, with $\mu$ playing the role of the external magnetic field. Therefore showing that hardcore Bose gas condenses on GSC graphs amounts to showing that the corresponding quantum XY model has a phase transition.

In finding a point of attack to this problem, we review results which are known to hold on arbitrary graphs. Given any infinite graph $G$, \cite{Cassi} showed that if the simple random walk on $G$ is recurrent, then for any $J>0$, the quantum XY model at zero field ($\mu=0$) on $G$ does not exhibit spontaneous magnetization. More precisely, if $\{G_n\}_n = \{(V_n, E_n)\}_n$ exhausts the recurrent graph $G$, and $\omega_{n,\beta,\mu}$ is the Gibbs state on $G_n$ at $(\beta,\mu)$, then for all $\beta>0$,
$$
\lim_{\mu\to 0}\lim_{n\to\infty} \frac{1}{|V_n|} \omega_{n,\beta,\mu}\left(\sum_{x\in V_n}S^1_x\right) =0.
$$
This result generalizes the Mermin-Wagner theorem~\cite{MerminWagner} from translationally invariant lattices to arbitrary graphs. Note that $\mu=0$ corresponds to bosons at half-filling (equal number of occupied and unoccupied vertices on average).

Specializing to the situation on GSC graphs, we deduce that on any GSC graph with $d_s\leq 2$, the quantum XY model at zero field does not undergo a phase transition. Under this setting, the hardcore Bose gas at half-filling does not Bose condense. Conversely, if condensation occurs on some GSC graph, then the graph is transient.

The harder problem is to prove sufficient conditions for which \emph{spontaneous magnetization},
$$
\lim_{\mu\to 0}\liminf_{n\to\infty} \frac{1}{|V_n|} \omega_{n,\beta,\mu}\left(\left|\sum_{x\in V_n}S^1_x\right|\right) > 0,
$$
or \emph{long-range order},
$$
\liminf_{n\to\infty} \frac{1}{|V_n|^2} \omega_{n,\beta,0}\left(\sum_{\langle xy\rangle \in E_n}\left(S_x^1 S_y^1 + S_x^2 S_y^2\right)\right) >0,
$$
occurs for large enough $\beta$. Either condition would imply the non-uniqueness of KMS states at low temperatures~\cite{BratelliRobinson}*{Section 6.2.6}.
\begin{problem}
Given any GSC graph $G$, prove or disprove that the transience of simple random walk on $G$ is a sufficient condition for the quantum XY model on $G$ to exhibit spontaneous magnetization or long-range order at positive temperatures. In other words, determine whether the hardcore Bose gas at half-filling condenses on any GSC graph with $d_s>2$.
\label{prob:QuantumXY}
\end{problem}

Recall that on $\mathbb{Z}^d$ this has been answered in the positive by~\cite{DysonLiebSimon}. However the use of Fourier transform was essential in their proof, which made the generalization to arbitrary graphs difficult. 

Though outside the scope of quantum mechanics, we also mention the abelian version of the XY model on general graphs, which is better understood. It has the same Hamiltonian~(\ref{eq:hcXY}), but with ${\bf S}_x=(S_x^1, S_x^2, S_x^3)$ taking values in the unit 2-sphere, \emph{i.e.,} $\sum_{j=1}^3 (S_x^j)^2=1$. It is known that there is no spontaneous rotational symmetry breaking at any temperature for recurrent graphs~\cite{Cassi,MerklWagner}. On the flip side, it appears that transience of simple random walk is a necessary, but not sufficient, condition for phase transition in the classical XY model. Y. Peres conjectured that the sufficient condition is that there exists a $p\in(0,1)$ such that upon performing an independent bond percolation with probability $p$, the resultant infinite cluster supports transient simple random walk~\cite{PemantleSteif}*{Conjecture 1.9}.

For GSC graphs, which possess strong symmetry and connectivity, we suspect that transience alone is a sufficient condition, though we are not aware of any previous work on this matter. The key may be to prove a version of spin-wave condensation in the spirit of~\cite{FrohlichSimonSpencer}.

\begin{problem}
Given any GSC graph $G$, prove or disprove that the transience of simple random walk on $G$ is a sufficient condition for the classical XY model on $G$ to exhibit spontaneous magnetization or long-range order at positive temperatures. 
\label{prob:ClassicalXY}
\end{problem}

\begin{acknowledgements}

I am grateful to Naotaka Kajino, Benjamin Steinhurst, Robert Strichartz, and Alexander Teplyaev for many useful discussions leading up to the writing of this work. Thanks also to Gerald Dunne for proposing the Sierpinski carpet waveguide example discussed in Section~\ref{sec:Casimir}; Michel Lapidus for suggesting the use of Ces\`{a}ro averaging in thermodynamical quantities; Daniele Guido for explaining his work on Bose-Einstein condensation on inhomogeneous graphs; and Matthew Begu\'{e} for providing the raw spectral data of the Kusuoka-Zhou Laplacian, from which Fig.~3\subref{subfig:KZSC} is generated. Portions of this work were presented at the "Waves and Quantum Fields on Fractals" workshop at Technion, and at the 6th Prague Summer School on Mathematical Statistical Physics. I wish to thank Eric Akkermans and Marek Biskup for organizing the respective conference and giving me the opportunity to speak.

During the preparation of the manuscript, the author was partially supported by Robert Strichartz's NSF grant (DMS 0652440), as well as a graduate assistantship from the 2011 Research Experience for Undergraduates (REU) program at the Department of Mathematics, Cornell University.

\end{acknowledgements}


\end{document}